\documentclass[preprint,12pt]{elsarticle}

\usepackage{latexsym, graphicx, epsfig, amsmath, amssymb, amsfonts}
\usepackage{amsmath, mathrsfs, mathtools, amsthm}
\usepackage{algorithm}
\usepackage{lineno}
\usepackage{algpseudocode}
\usepackage{amssymb}
\usepackage{graphicx}
\usepackage{subfig}
\usepackage{tabu}
\usepackage{adjustbox}
\usepackage{bm}
\usepackage{comment}
\usepackage{tikz}
\usepackage{hyperref}
\usetikzlibrary{shapes.geometric}
\usetikzlibrary{arrows.meta,calc}
\usepackage{xcolor} 
\usepackage{multirow}
\theoremstyle{plain}
\newtheorem{theorem}{Theorem}[section]
\newtheorem{lemma}[theorem]{Lemma}
\newtheorem{remark}{Remark}[section]
\numberwithin{figure}{section}
\numberwithin{table}{section}
\usetikzlibrary{positioning}
\usepackage{xcolor}
\definecolor{darkgreen}{RGB}{0,100,0}    

\definecolor{darkyellow}{RGB}{204,153,0}  
\newcommand{\jd}[1]{{\color{black}#1}}
\newcommand{\jdd}[1]{{\color{black}#1}}
\newcommand{\jddd}[1]{{\color{black}#1}}
\newcommand{\jdall}[1]{{\color{black}#1}}

\makeatletter
\def\ps@pprintTitle{%
	\let\@oddhead\@empty
	\let\@evenhead\@empty
	\let\@oddfoot\@empty
	\let\@evenfoot\@empty
}
\makeatother
\begin{document}

\title{Laplacian Eigenfunction-Based Neural Operator for Learning Nonlinear Reaction–Diffusion Dynamics}
\author[l1]{Jindong Wang}
\ead{jzw6472@psu.edu}
\author[l1]{Wenrui Hao\corref{cor}}
\ead{wxh64@psu.edu}
\cortext[cor]{Corresponding author}

\affiliation[l1]{organization={Department of Mathematics, Penn State University},
            city={University Park},
            postcode={16802}, 
            state={PA},
            country={USA}}

\begin{abstract}
Learning reaction--diffusion equations has become increasingly important across scientific and engineering disciplines, including fluid dynamics, materials science, and biological systems. In this work, we propose the \textit{Laplacian Eigenfunction-Based Neural Operator} (LE-NO), a novel framework designed to efficiently learn nonlinear reaction terms in reaction--diffusion equations. LE-NO models the nonlinear operator on the right-hand side using a data-driven approach, with Laplacian eigenfunctions serving as the basis. This spectral representation enables efficient approximation of the nonlinear terms, reduces computational complexity through direct inversion of the Laplacian matrix, and alleviates challenges associated with limited data and large neural network architectures---issues commonly encountered in operator learning. We demonstrate that LE-NO generalizes well across varying boundary conditions and provides interpretable representations of learned dynamics. Numerical experiments in mathematical physics showcase the effectiveness of LE-NO in capturing complex nonlinear behavior, offering a powerful and robust tool for the discovery and prediction of reaction--diffusion dynamics.
\end{abstract}

\maketitle

\section{Introduction}\label{sec:intro}

With the rapid growth of data and advances in machine learning, data-driven discovery of partial differential equations (PDEs) has gained significant attention. Unlike traditional approaches that derive physical laws from first principles, this framework uncovers governing equations directly from observational data, which is particularly appealing in fields like biology and medicine where complex interactions and incomplete mechanistic understanding often hinder accurate model development \cite{ge2024data}.

Early methods for PDE discovery, such as sparse regression \cite{tibshirani1996regression}, sequential threshold ridge regression \cite{rudy2017data}, SINDy \cite{brunton2016discovering}, \jddd{and regression approaches for nonlinear integro–differential operators \cite{patel2018nonlinear}}, demonstrated the feasibility of recovering canonical PDEs in physics. However, these approaches are often limited to systems where the governing equations contain relatively simple terms. As the complexity of the system increases, particularly in high-dimensional biological systems, the computational cost of searching large symbolic libraries becomes prohibitive, and symbolic regression methods may fail to produce accurate models due to the vast number of potential terms.

Recent developments have shifted towards symbolic regression methods that do not rely on fixed libraries \cite{kim2020integration,lu2022discovering,zhang2023deep,vaddireddy2019equation,zhang2025coefficient}. These methods have demonstrated superior performance in recovering partial differential equations (PDEs) from noisy and sparse data. However, they come with their own set of challenges. Most existing approaches have been tested on synthetic datasets with known PDEs, which limits their applicability to real-world systems where the governing equations are often unknown or nontrivial.

Biological systems, in particular, pose unique challenges for the discovery of PDEs. For example, reaction--diffusion models frequently used in biology often feature reaction terms that vary across individuals or are difficult to express in closed form. These limitations of traditional methods motivate the need to move beyond symbolic regression and toward a more flexible operator-learning framework. Deep learning-based methods such as physics-informed neural networks (PINNs) \cite{raissi2017physics,raissi2019physics,siegel2023greedy} and operator learning approaches \cite{hao2024newton,yu2024nonlocal,zhang2024modno,liu2024deep} are emerging as promising alternatives.

In this paper, \textbf{Laplacian Eigenfunction-Based Neural Operator (LE-NO)} introduces a novel way of discovering nonlinear PDEs and computing their solutions simultaneously. Unlike traditional operator learning methods, which rely on learning the PDE term-by-term, LE-NO approximates the entire nonlinear operator directly by integrating Laplacian eigenfunctions into a neural network architecture. These eigenfunctions, derived from the underlying PDE operator, provide a natural physics-aligned basis that simplifies the learning task, ensures interpretability, and enhances computational efficiency. The Laplacian eigenfunctions diagonalize the diffusion term, allowing the network to focus on learning the nonlinear part of the operator, which leads to a more stable and accurate solver. This is a departure from Fourier Neural Operators (FNOs), which are parameterized in the frequency domain and are primarily designed for periodic boundary conditions. LE-NO, by contrast, is flexible enough to handle non-periodic boundary conditions and irregular geometries, where traditional Fourier modes may fail.

One of the primary contributions of LE-NO lies in its ability to handle complex, heterogeneous systems, such as those seen in biology. By learning the operator directly from data, LE-NO circumvents the need for predefined symbolic terms or the integration of physical laws, offering a more adaptable and scalable solution to PDE discovery in challenging biological systems. Moreover, LE-NO extends beyond typical operator learning frameworks by enabling solution approximation, time stepping, and transfer learning across individuals with varying disease states or parameters, making it ideal for applications like personalized Alzheimer’s disease modeling.

In addition to discovering the right-hand side of a PDE, LE-NO can also be used to learn the solution operator, construct neural solvers for time stepping, and perform transfer learning across patients. These capabilities go beyond what current symbolic regression methods can achieve and more closely align with recent advances in operator learning, such as DeepONet \cite{lu2019deeponet,lu2021learning} and FNO \cite{li2020fourier}. 

One of LE-NO’s distinguishing advantages is its enhanced interpretability. By projecting the solution space onto a set of Laplacian eigenfunctions—natural modes of the underlying differential operator—the learned dynamics can be more transparently understood in terms of their spectral contributions. This basis aligns closely with the physics of diffusion-type processes and often corresponds to spatial patterns observed in biological systems. This interpretability is particularly valuable in biological applications like Alzheimer’s disease modeling, where understanding the spatial and temporal progression of biomarkers can inform hypotheses, guide experimental design, and build trust in AI-generated predictions \cite{pang2023geometric}.

To validate LE-NO’s performance, we compare it against established methods like DeepONet and FNO on canonical PDEs, showing that LE-NO achieves superior efficiency and accuracy, especially when dealing with sparse and noisy data. Theoretical guarantees on approximation and convergence are also provided, ensuring that LE-NO performs reliably across a range of applications. Additionally, we apply LE-NO to Alzheimer’s disease biomarker data, where it uncovers previously unrecognized nonlinear dynamics in disease progression, outperforming symbolic regression and demonstrating its potential for real-world biological applications.

\section{Problem Statement}\label{sec:preliminary}

We consider the following general nonlinear reaction--diffusion dynamic:
\begin{equation}\label{eq:baseeq}
	\left\{
	\begin{aligned}
		u_t-D\Delta u&= \mathcal{F}(u), & \text{ in }\Omega\times [0,T]\\
		u(x,0) &= u_0(x), &\text{ in }\Omega,
	\end{aligned}\right.
\end{equation}
with various boundary conditions including Dirichlet or Neumann boundary condition. $D$ represents the diffusion coefficient and $\Omega\subset \mathbb{R}^d$ is a polygonal domain, and $\mathcal{F}$ denotes the nonlinear term determined by empirical data.

The weak form of problem \eqref{eq:baseeq} with homogeneous Dirichlet boundary condition reads as: Find $u\in L^2(0,T;H_0^1(\Omega))$ with $u_t\in L^2(0,T;H^{-1}(\Omega))$ such that
\begin{equation}\label{eq:originvp}
	(u_t,v)_{L^2}+(D\nabla u,\nabla v)_{L^2} = (\mathcal{F}(u),v)_{L^2} 
\end{equation}
for $ v \in H^1_0(\Omega)$ with $u(x,0) = u_0(x)$.

The nonlinear term  $\mathcal{F}$  can also be viewed as a nonlinear operator:
\begin{equation}
	\mathcal{F}:H^1(\Omega) \rightarrow H^{-1}(\Omega). 
\end{equation} 

The primary objective of this paper is the data-driven discovery of the nonlinear term $\mathcal{F}$  represented by a neural operator, $\mathcal{N}(u;\theta)$ parameterized by $\theta$. This framework is general enough to encompass a wide range of problems in mathematical physics, including pattern formation in biology \cite{gierer1972theory}, population dynamics \cite{holmes1994partial}, catalytic reactions \cite{gupta2009linear}, and diffusion in alloys \cite{allen1975coherent}. By leveraging the neural operator, we aim to develop a data-driven model that facilitates efficient prediction and enables the discovery of causal networks underlying their behavior.

\section{Proposed Methodology}
This section presents an overview of the LE-NO framework, as illustrated in Figure~\ref{fig:nn}, together with the methodology and theoretical analysis of the approach.
\subsection{Neural Operator Framework}
Given Hilbert spaces $\mathcal{X}(\Omega)$, $\mathcal{Y}(\Omega)$, we define the following neural operator
\begin{equation}\label{eq:no1}
	\mathcal{N}({u}) = \sum_{i=1}^P \mathcal{A}_i \sigma\left(\mathcal{W}_i {u} + \mathcal{B}_i\right) \quad \forall {u} \in \mathcal{X},
\end{equation}
where
$\mathcal{W}_i \in \mathcal{L}(\mathcal{X}, \mathcal{X})$: continuous linear operator,
$\mathcal{B}_i \in \mathcal{X}$: bias term, $\mathcal{A}_i \in \mathcal{L}(\mathcal{X}, \mathcal{Y})$: continuous linear operator, and $\sigma: \mathbb{R} \to \mathbb{R}$: nonlinear pointwise activation function.

Normally, the linear operator \(\mathcal{A}_i\) is composed of neural network basis functions, i.e., 
\begin{equation}\label{eq:no2}
	\mathcal{A}_iv = \sum_{j=1}^m {A}_i^j(v) \sigma(\alpha_i^j \cdot x + \zeta_i^j),\quad x\in\Omega,\ v\in\mathcal{Y}.
\end{equation}
where $A_i^j:\mathcal{X}\rightarrow\mathbb{R}$, $\alpha_i^j\in\mathbb{R}^d$ and $\zeta_i^j\in \mathbb{R}$. This formulation is designed to learn low-dimensional intrinsic mappings \cite{lu2019deeponet, li2020fourier, he2023mgno}.
Since the operator learning involves two neural network representations, the training process becomes quite challenging due to the limited data available, and the architecture itself is also very large. Therefore, we consider using the eigenfunctions of the Laplacian  instead of the neural network basis $ \sigma(\alpha_i^j \cdot x + \zeta_i^j)$ in Equation~\ref{eq:no2}, as they offer good approximation properties and are well-suited for domains with complex geometries.

\begin{figure}[!htbp]
	\centering
	\resizebox{1\textwidth}{!}{
		\definecolor{darkblue}{rgb}{0.0, 0.0, 0.55}
		\begin{tikzpicture}[
			every node/.style={align=center},
			circ/.style={draw, circle, thick, fill=gray!20, inner sep=2pt}, 
			box/.style={draw=darkblue, rounded corners, align=center, minimum width=2cm, minimum height=2cm,line width=.03cm},
			dashedbox/.style={draw, dashed, rounded corners, minimum width=12.5cm, minimum height=3.2cm, line width=.03cm},
			arrow/.style={thick,->,>=stealth},
			bluearrow/.style={-{Triangle[angle=60:5pt,fill=blue]}, thick},
			dashedline/.style={dashed, line width=0.5mm},
			innerbox/.style={draw, rounded corners, align=center, minimum width=2cm, minimum height=1cm},
			box2/.style={draw=darkblue, rounded corners, align=center, minimum width=4cm, minimum height=2cm, thick, line width=1.5mm},
			]
			
			\node[box] (data) {1. Data\\$u^0(x),t_0$\\$u^1(x),t_1$\\$\vdots$\\$u^N(x),t_N$};
			\node[dashedbox, right=of data] (offline) {2a. Offline Computation\\\\\\\\\\\\\\};
			\node[right=of data] (eigenp) {Eigen Problem \\$-\Delta \phi_i=\lambda_i\phi_i$};
			\node[box, right=.5cm of eigenp] (eigen) {Eigenfunction\\$(\lambda_1,\phi_1)$\\$(\lambda_2,\phi_2)$\\$\vdots$\\$(\lambda_P, \phi_P)$};
			\node[box, right=2.2cm of eigen] (beta){Coefficient~Residual\\$\boldsymbol{\beta}^0(u)\quad\mathcal{R}^0(u)$\\$\boldsymbol{\beta}^1(u)\quad\mathcal{R}^1(u)$\\$\vdots\quad\qquad\vdots$\\$\boldsymbol{\beta}^N(u)\quad\mathcal{R}^N(u)$};
			\node[right=0cm of eigen, yshift=.3cm] (label) {Projection};
			\node[box, below=1.5cm of data, xshift=.5cm] (net) {2b. Neural Network $\mathcal{G}$\\\includegraphics[width=.25\textwidth]{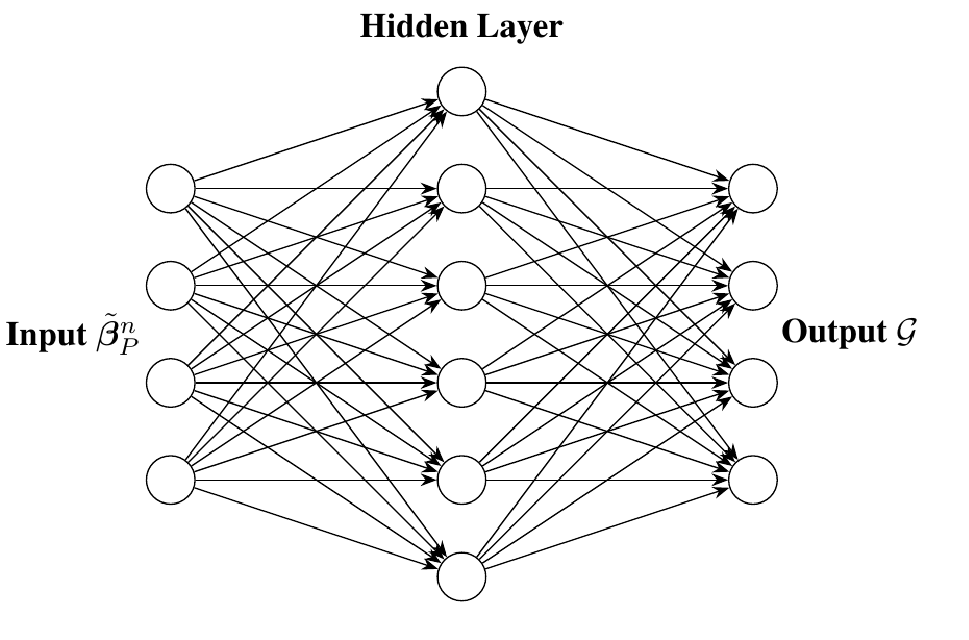}};
			\node[circ,right=.5cm of net] (time) {$\times$};
			\node[below=of data, xshift=-1.2cm] (start) {};
			\node[right=14.8cm of start] (end) {};
			\node[box, right=.5cm of time] (no) {3a. Neural Operator \\\parbox{4.8cm}{$$\mathcal{N}(u;\theta)=\sum_{i=1}^{P}\mathcal{G}_i(\boldsymbol{\beta}(u);\theta)\phi_i$$}};
			\node[box,right=.5cm of no] (ite) {3b. PDE Evolution\\\\$\widetilde{\bm{\beta}}^0\rightarrow\widetilde{\bm{\beta}}^1\rightarrow\cdots\rightarrow\widetilde{\bm{\beta}}^N$};
			\node[box,right= 1cm of ite, yshift=2.35cm] (loss) {4. Solve Optimization\\ \parbox{1.5cm}{$$
					\underset{\theta}{\arg \min } ~L(\theta)$$}};
			\draw[bluearrow] (data) -- (offline);
			\draw[bluearrow] (eigenp)--(eigen);
			\draw[bluearrow] (eigen)--(beta);
			\draw[dashedline] (start)--(end);
			\draw[arrow] (net)--(time);
			\draw[arrow] (eigen)--(time);
			\draw[arrow] (time)--(no);
			\draw[bluearrow] (no)--(ite);
			\draw[bluearrow] (ite)-| (loss);
			\draw[bluearrow] (beta)-| (loss);
	\end{tikzpicture}}
	\caption{Framework of PDE learning using the Laplacian Eigenfunction-Based Neural Operator (LE-NO). The process involves computing Laplacian eigenfunctions on the domain \(\Omega\), projecting observed data onto these eigenfunctions to obtain \(\bm{\beta}^n\) and the residual \(\mathcal{R}^n\), integrating the eigenfunctions into the neural operator to approximate the nonlinear  term of PDEs, and training the neural operator using a loss function based on the PDEs evolution.}
	\label{fig:nn}
\end{figure}

\subsection{Laplacian Eigenfunction}
In this paper, instead of using neural network basis functions, we leverage the eigenfunctions of $-\Delta$. Specifically, using the inhomogeneous Dirichlet boundary condition as an example, we consider the following eigenvalue problem:

\begin{equation}
	\left\{	\begin{aligned}
		-\Delta \phi_i &= \lambda_i \phi_i, &\text{ in }\Omega,\\
		\phi_i &= 0, & \text{ on }\partial \Omega,
	\end{aligned}\right.
\end{equation}
with $\|\phi\|_{L^2(\Omega)} = 1$ and it holds $0\le\lambda_1\le \lambda_2\le \cdots$. 
\begin{remark}
	For a rectangular or cubic domain, the eigenfunctions are composed of tensor products of Fourier modes. On an irregular domain, however, eigenfunctions can be derived using discretization methods such as finite element methods. This derivation is performed as an offline preprocessing step to train the neural operator efficiently. 
\end{remark}

Thus, the \textbf{Laplacian Eigenfunction-Based Neural Operator} is defined as 
\begin{equation}
	\mathcal{N}({u}) = \sum_{i=1}^P {A}_i \sigma\left(\mathcal{W}_i {u} + \mathcal{B}_i\right)\phi_i(x) \quad \forall {u} \in \mathcal{X}, \label{LENO}
\end{equation}
where $A_i:\mathcal{X}\rightarrow\mathbb{R}$. It is employed to approximate $\mathcal{F}$ in Eq. (\ref{eq:baseeq}) to learn the exact form of the model.

\jddd{
\begin{remark}
The Laplacian eigenfunctions were also employed in \cite{chen2024learning} to construct an encoder–approximation–decoder block: the encoder projects data onto the eigenfunction basis via an $L^2$ projection, and the decoder reconstructs the signal using the same basis. This block is incorporated as a nonlinear layer within their neural operator, playing a role analogous to the Fourier kernel module in the Fourier Neural Operator. By contrast, our approach can be viewed as a variant of DeepONet, in which the trunk network is replaced by Laplacian eigenfunctions, while the branch network remains a linear functional approximated by a neural network. This distinction highlights the fundamental difference between the two methods. Furthermore, although our paper focuses on problems defined over a domain $\Omega$, our method naturally extends to manifolds, since the corresponding Laplacian eigenfunctions can also be computed.   
\end{remark}

}
\subsection{Learning Problem}
For the learning problem, we will project both the data and the equation onto the following space:
\begin{equation}
	V_P = {\rm span}\{\phi_i\}_{i=1}^P.
\end{equation}

{\bf Data projection.}  Given the observed data consisting of data samples $\{u^n(x), t^n\}_{n=0}^N$ at time $t_n$, the projection is defined as: $\beta_i^n(u)=(u^n(x),\phi_i(x))$, $i=1,\cdots,P$.

{\bf Equation projection.} By restricting the variational problem \eqref{eq:originvp} to the subspace $V_P$, we approximate the numerical solution as 
$\displaystyle \widetilde{u}^n(x)=\sum_{i=1}^{P}\widetilde{\beta}_i^n\phi_i(x).$

By applying the semi-implicit Euler method for time discretization and replacing  $\mathcal{F}$  with  $\mathcal{N}$, we derive the following discrete formulation
\jddd{\begin{equation}\label{eq:discrete}
	\frac{\widetilde{\bm{\beta}}^n-\widetilde{\bm{\beta}}^{n-1}}{ t_n-t_{n-1}}+D\Lambda_P\widetilde{\bm{\beta}}^n =\begin{pmatrix}
	    {A}_1 \sigma\left(\mathcal{W}_1 \widetilde{u}^{n-1} + \mathcal{B}_1\right)\\
        {A}_2 \sigma\left(\mathcal{W}_2 \widetilde{u}^{n-1} + \mathcal{B}_2\right)\\
        \vdots\\
        {A}_P \sigma\left(\mathcal{W}_P \widetilde{u}^{n-1} + \mathcal{B}_P\right)\\
	\end{pmatrix} , \quad  n=1,\cdots,N,
\end{equation}}
where $\bm{\widetilde{\beta}}^n=(\widetilde{\beta}_1^{n},\cdots,\widetilde{\beta}_P^{n})^T$ and $\Lambda_P = {\rm diag}(\lambda_1,\cdots,\lambda_P)$ (due to the orthogonality of $\phi_i$). The nonlinear component of the neural operator, originally expressed as ${A}_i\sigma(\mathcal{W}_i\widetilde{u}^{n-1}+\mathcal{B}_i)$ can now be regarded as a continuous function $\mathcal{G}_i(\widetilde{\bm\beta}^{n-1}(u);\theta)$ which depends on the numerical approximation $\widetilde{\bm \beta}$ of solution coefficient $\bm\beta=(\beta_1,\beta_2,\cdots,\beta_P)^T$ on the eigenfunction space and and is parameterized by a neural network.  We denote the collection as $\mathcal{G}=(\mathcal{G}_1,\cdots,\mathcal{G}_P)^T$. For simplicity, the dependence of $\bm{\widetilde{\beta}}$ on $u$ is omitted for simplicity.

\jddd{
\begin{remark}
Equation \eqref{LENO} illustrates the single-layer structure of the neural operator, which can be viewed as a generalization of a shallow neural network and provides intuition for our approach. In this formulation, the shallow representation corresponds to the functional component of the neural operator, which in practice can be replaced by a deep neural network. In other words, the mapping $\mathcal{G}$ can be realized with a multi-layer architecture to increase expressive power, while the basis functions remain the Laplacian eigenfunctions, ensuring established approximation properties. This is the strategy adopted in our implementation.
\end{remark}

}

{\bf Training Loss.} 
For simplicity, we denote both the time differentiation operator and the Laplace operator as:
\begin{equation}
	\mathcal{R}^n(u) = \frac{\bm{\beta}^n(u)-\bm{\beta}^{n-1}(u) }{t_n-t_{n-1}} + D\Lambda_P \bm{\beta}^n(u),
\end{equation}
with $\bm\beta^n(u)=(\beta_1^n(u),\beta_2^n(u),\cdots,\beta_n^P(u))^T$. To train the neural network $\mathcal{G}(\cdot,\theta):\mathbb{R}^P\rightarrow\mathbb{R}^P$, we focus on both the data RMSE loss and the model loss:
\begin{itemize}
	\item {\textbf{Data $L^2$ loss}
		ensures that the discretized solution $\bm{\widetilde{\beta}}^n$ closely matches the given data $\bm{\beta}^n$     
	}:
	\begin{equation}
		L^D(\theta):= \frac{1}{N}\sum_{n=1}^{N} \frac{\|\bm{\widetilde{\beta}}^n-\bm{\beta}^n\|_{l^2}}{\|\bm{\beta}^n\|_{l^2}}.
	\end{equation}
	\item 	{\textbf{Model residual loss}  guarantees that the solution satisfies the model problem}:
	\begin{equation}
		L^R(\theta):=\frac{1}{N}\sum_{n=1}^{N} \frac{ \|\mathcal{R}^n-\mathcal{G}(\bm{{\beta}}^{n-1};\theta)\|_{l^2}}{\|\mathcal{R}^n\|_{l^2}}.
	\end{equation}
\end{itemize}

Then we train the neural operator by minimizing the following combined loss function:
\begin{equation}
	\min_\theta L(\theta):= L^D(\theta)+L^R(\theta).
\end{equation}

\jddd{
\begin{remark}
    The two loss terms, $L^D$ and $L^R$, are both defined using the relative $l^2$ error. Since this normalization ensures that the two terms are of comparable magnitude, it is not necessary to introduce additional weighting, and we simply use their summation as the total loss.
\end{remark}

}
Thus we summarize the detailed implementation in Algorithm 1 for learning the nonlinear term 
$\mathcal{F}$ of the PDE.

\begin{algorithm}[!htbp]
	\caption{Nonlinear Term Learning for PDE Solving}
	\label{alg:no}
    \jd{
	\begin{algorithmic}[1]
		\Require 
        Dynamic data with $M$ samples $\{(u_i^n(x), t_n)\}_{n=0}^N$ for $i=1,\dots,M$ on domain $\Omega$;  
        integer $P$ (number of eigenfunctions);  
        integer $epochs$ (number of training epochs);  
        randomly initialized neural network $\mathcal{G}(\cdot;\theta^0):\mathbb{R}^P \to \mathbb{R}^P$ with parameters $\theta^0$.

		\State Compute Laplacian eigenfunctions $\{\phi_i\}_{i=1}^P$ on $\Omega$.  
        \State Project the data samples to obtain coefficients $\bm{\beta}_i^n$ and residuals $\mathcal{R}_i^n$:  
        \[
        \bm{\beta}_i^n = \big( (u_i^n,\phi_1),\,(u_i^n,\phi_2),\,\dots,\,(u_i^n,\phi_P) \big)^\top,
        \]
        \[
        \mathcal{R}_i^n = \frac{\bm{\beta}_i^n - \bm{\beta}_i^{n-1}}{t_n - t_{n-1}} + D \Lambda_P \bm{\beta}_i^n.
        \]

        \For{$j$=$1,2,\cdots,epochs$}
		\For{$n = 1, 2, \ldots, N$}
		\State Evolve the approximation $\widetilde{\bm{\beta}}_i^n$ with initial condition $\widetilde{\bm{\beta}}_i^0 = \bm{\beta}_i^0$:  
            \[
            \frac{\widetilde{\bm{\beta}}_i^n - \widetilde{\bm{\beta}}_i^{n-1}}{t_n - t_{n-1}} 
            + D\Lambda_P \widetilde{\bm{\beta}}_i^n
            = \mathcal{G}(\widetilde{\bm{\beta}}_i^{\,n-1}; \theta^{j-1}).
            \]
        \EndFor
        
		\State Use the loss function $L$:  
        \[
        L(\theta^{j-1}) = \frac{1}{MN}\sum_{i=1}^M\sum_{n=1}^N 
        \left( 
        \frac{\|\widetilde{\bm{\beta}}_i^n - \bm{\beta}_i^n\|_{2}}{\|\bm{\beta}_i^n\|_{2}}
        + \frac{\|\mathcal{R}_i^n - \mathcal{G}(\bm{\beta}_i^{\,n-1})\|_{2}}{\|\mathcal{R}_i^n\|_{2}}
        \right).
        \]

        \State Update parameters using gradient descent:  
        \[
        \theta^j = \theta^{j-1} - \eta \nabla L(\theta^{j-1}).
        \]
        \EndFor

		\Ensure Nonlinear learning term of the PDE:  
    \[
    \mathcal{N}(\cdot) = \sum_{i=1}^P \mathcal{G}_i(\cdot;\theta^*) \, \phi_i(x), 
    \qquad \theta^* = \theta^{epochs}.
    \]
        
	\end{algorithmic}}
\end{algorithm}

As illustrated in Figure \ref{fig:nn}, the framework is specifically designed to learn the complex nonlinear operator $\mathcal{F}$.

\jdall{
\begin{remark}[Extension to general diffusion]
Our method also extends to general diffusion operators of the form
\begin{equation}
u_t - \nabla \cdot \big(D(x)\nabla u\big) = \mathcal{F}(u),
\end{equation}
where $D(x)$ may be either a strictly positive scalar function or a symmetric positive definite (SPD) matrix. In this case, one considers the eigenvalue problem
\begin{equation}
-\nabla \cdot \big(D(x)\nabla \phi_i\big) = \lambda_i \phi_i,
\end{equation}
whose eigenfunctions still satisfy $L^2$ orthogonality. Hence, the proposed method applies directly: it suffices to replace $D\Lambda_P$ by $\Lambda_P$ in the algorithm, as the effect of the diffusion coefficient is already encoded in the eigenproblem.
\end{remark}
}

\begin{remark}
	Our method is also applicable to inhomogeneous Dirichlet and Neumann boundary conditions. For detailed discussions and adaptations, refer to  \ref{appendix:boundary}.
	
\end{remark}

\subsection{Theoretical Analysis}
In this section, we aim to analyze the convergence and the approximation of the Laplacian Eigenfunction-Based Neural Operator. 
The following theorem provides the approximation property for $\mathcal{N}$.
\begin{theorem}[\bf operator approximation error]
	\label{thm:approxop}
	Let $\mathcal{F}: \mathcal{X}(\Omega)\rightarrow \mathcal{Y}(\Omega)$ be a continuous operator with $\|\mathcal{F}'\|\le L$, $K\subset \mathcal{X}\cap H^k(\Omega)$ be a compact set and $\mathcal{F}(K)\subset H^k(\Omega)$, there exists a neural operator $\mathcal{N}$ of the form \eqref{LENO}, where $\mathcal{G}$ is a shallow neural network with width $m$, such that
	\begin{equation}\label{eq:approxop}
		\begin{aligned}
			\sup_{u\in K} &\|\mathcal{N}(u)-\mathcal{F}(u)\|_{L^2(\Omega)}\\
			&\le C(L+1)P^{-{\frac{k}{d}}}+C\log(m)^{{\frac{1}{2}}+P}m^{-{\frac{1}{P}}{\frac{P+2}{P+4}}}.
		\end{aligned}
	\end{equation}
	where $C$ is a constant depending on the compact set $K$,
\end{theorem}

\begin{remark}
	Note that for a shallow neural network, given a fixed input dimension $P$, we can set its width to be $\mathcal{O}(e^{P^\alpha})$ for some $\alpha > 2$ in order to achieve a small approximation error when $P$ is large.

	Alternatively, we can use a deep neural network by setting the depth as $L=\mathcal{O}(\frac{k}{d}\log(P)+1)$ and the number of parameters as $m=\mathcal{O}\big(e^{\frac{kP}{d}\log{P}}(\frac{k}{d}\log(P)+1)\big)$ to give the following estimate, 
	\begin{equation}
		\begin{aligned}
			\sup_{u\in K} &\|\mathcal{N}(u)-\mathcal{F}(u)\|_{L^2(\Omega)}
			\le CP^{-{\frac{k}{d}}},
		\end{aligned}
	\end{equation}
	which is derived based on the $L^\infty$ estimate of deep ReLU neural networks from \cite{yarotsky2017error}. Here, the exponential dependency in $m$ exists due to the general Lipschitz assumption on $\mathcal{F}$.
\end{remark}
Under this result, we have the following estimate for the approximated solution $\widetilde{u}^n$ obtained by the neural operator.
\begin{theorem}[\bf solution approximation error]
	\label{thm:err-sol}
	Under assumption in Theorem \ref{thm:approxop}, let $u(t)$ be the solution of \eqref{eq:baseeq} and $\widetilde{u}^n$ be the solution of \eqref{eq:discrete}, there exists a neural operator $\mathcal{N}$ such that
	\begin{equation}\label{eq:err-sol}
		\begin{aligned}
			\sup_{0\le n\le N}	\| u(t_n) - \widetilde{u}^n\|_{L^2(\Omega)} \le &  
			C\sup_{u\in K}\|\mathcal{F}(u)-\mathcal{N}(u)\|_{L^2} T^{1/2}\exp(CLT)\\
			&~ +C(\tau+P^{-{\frac{k}{d}}}) .
		\end{aligned}
	\end{equation}
	where $\widetilde{u}^n = \sum_{i=1}^P \widetilde{{\beta}}^n_i\phi_i$, $\tau=\max_n(t_n-t_{n-1})$ and $\sup_{u\in K}\|\mathcal{F}(u)-\mathcal{N}(u)\|_{L^2}$ can be bounded by the result in Theorem \ref{thm:approxop}.
\end{theorem}

\begin{remark}
	In practice, the available data may have limited regularity, such as belonging to $L^\infty(\Omega)$. By the Sobolev embedding theorem  \cite{adams2003sobolev}, it holds that $L^\infty(\Omega) \hookrightarrow H^k(\Omega)$ for $k<{d\over 2}$, 
	thus the term $P^{-{k\over d}}$ in \eqref{eq:approxop} and \eqref{eq:err-sol} can be replaced by
	
	\begin{equation}
		P^{-{k\over d}} \rightarrow P^{-{1\over 2}+\xi}, 
	\end{equation}
	for arbitrary $\xi >0$. For a uniform time step $\tau = \frac{T}{N}$, with sufficiently large $m$, $N$ and $P$, the error can be made arbitrarily small.
\end{remark}

We explore the convergence behavior of the Laplacian Eigenfunction-Based Neural Operator by analyzing the neural tangent kernel (NTK) \cite{jacot2018neural}, which provides estimates of the evolution of the parameter $\theta(s)$ with respect to $s$, as detailed in the  \ref{appdendix:opt}.

{We present convergence results for two cases: (1) optimization using only the residual loss $L^R$ over multiple time steps, and (2) optimization using both the residual loss $L^R$ and the data loss $L^D$ in a one-step setting, due to the nested structure of neural networks inherent in multi-step time evolution.}

\begin{theorem}[\bf convergence of training error]
\label{thm:opt}
Let $\mathcal{G}(\cdot;\theta)$ be the shallow neural network trained using the absolute mean squared loss $L$ in Algorithm \ref{alg:no}, with $M$ data samples and $N$ evolution steps. In the infinite-width limit $m \rightarrow \infty$, there exist constants $s^*>0$ and $\gamma > 0$ such that the loss satisfies 
$$L(\theta(s)) \le L(\theta(0))\exp(-s\lambda),$$
for all $s \in [0, s^*]$ with high probability. The convergence rate $\lambda$ depends on the training setup and is characterized by the following cases:
\begin{itemize}
\item[{(a)}] {General $N$-step training using only residual loss $L=L^R$:}
$$\lambda = \frac{4\gamma}{NM}, $$
\item[{(b)}] {One-step evolution ($N=1$) with combined loss $L = L^R + L^D$ as in Algorithm~\ref{alg:no}:}  
$${\lambda}={4\gamma\over M}\bigg(1+2 \frac{\tau(1+\tau\lambda_P)^{-1}}{(1+\tau^2(1+\tau\lambda_1)^{-2})}\bigg).$$
\end{itemize}
\end{theorem}

\begin{remark}
	The smallest eigenvalue of the kernel, denoted by $\lambda$ plays a key role in determining the convergence rate. When training with only the residual loss $L^R$ under a one-step evolution ($N = 1$), the resulting convergence factor is smaller than that obtained using the combined loss $L = L^R + L^D$. This suggests that incorporating both residual and data losses can lead to accelerated convergence.
\end{remark}

\section{Numerical Examples}\label{sec:experiment}
In this section, we evaluate the effectiveness of the proposed approach by examining several PDE problems in both 1D and 2D cases. Additionally, we compare the performance of the proposed Laplacian Eigenfunction-Based Neural Operator against several existing neural operators, including FNO and DeepONet. For detailed information on the datasets and the training process, please refer to  \ref{appdendix:test}.

The training accuracy is evaluated using three different metrics for $M$ data samples, as defined below:
\begin{itemize}
	\item Relative $L^2$  error of the solution (denoted by $E_{L^2}$ ): 
	$${1\over MN}\sum_{m=1}^M\sum_{n=1}^{N} {\|{u}^n_m-u_m(t_n)\|_{L^2}\over \|u_m(t_n)\|_{L^2}}$$
	\item Relative PDE residual error (denoted by $E_{Res}$):
	$${1\over MN}\sum_{m=1}^M\sum_{n=1}^{N} {\|\mathcal{G}^\theta(\bm{{\beta}}^{n-1}(\widetilde{u}^n_m))-\mathcal{R}^n(u_m)\|_{l^2}\over \|\mathcal{R}^n(u_m)\|_{l^2}}$$
	\item Relative nonlinear term error (denoted by $E_{Nonlinear}$):
	$${1\over MN}\sum_{m=1}^M\sum_{n=1}^{N} {\|\mathcal{N}(\widetilde{u}^n_m)-\mathcal{F}(u_m(t_n))\|_{L^2}\over \|\mathcal{F}(u_m(t_n))\|_{L^2}}$$
\end{itemize}

These metrics provide a comprehensive assessment of the proposed learning approach’s performance in approximating the solution, satisfying the PDE constraints, and accurately capturing the nonlinear terms.
\subsection{Example 1: KPP-Fisher  Equation}\label{subsec:kpp}
The KPP-Fisher equation is a classic benchmark in reaction--diffusion models, commonly used to describe population growth and wave propagation \cite{holmes1994partial,fisher1937wave}. Its mathematical form in 1D case is given by\begin{equation}
	\left\{	\begin{aligned}
		{\partial u\over \partial t}-{\partial^2 u\over \partial x^2}  &= u(1 -u),  &x\in (0,1),t\in (0,T]\\
		u(x,0) &= u_0(x), &x\in (0,1)
	\end{aligned}\right.
\end{equation}
with homogeneous Dirichlet boundary conditions where $u_0 \in L^2((0,1);\mathbb{R})$ is the initial condition. 

\paragraph{Comparison with  Various Methods}

The performance of our proposed method, along with comparisons to FNO and DeepONet, is summarized in Table \ref{tab:kpp-compare}. We use $P = 64$ eigenfunctions to solve the discrete problem \eqref{eq:discrete}. The DeepONet features a branch network with three hidden layers and a trunk network with two hidden layers, each with a width of 1000 For the FNO structure used for $\mathcal{N}$, we employ a 1D FNO with four integral layer, 64 Fourier modes, and a width of 64. Additionally, the loss curves for different methods are shown in Figure \ref{fig:kpp-loss}. It is noteworthy that the neural operators used in this comparison have parameters of a similar scale.

\begin{table}[!htbp]
	\centering
	\caption{Comparison of errors between \textbf{multilayer} LE-NO, DeepONet, and FNO on the 1D KPP–Fisher equation.}
	\begin{tabular}{c|ccccc}
		\hline
		Method &  $E_{L^2}$ &  $E_{Res}$  &$E_{Nonlinear}$ & Time/Epoch & \#Parameter\\
		\hline
		LE-NO&5.54e-05&1.16e-03&1.17e-03 & 0.03s & 1.13M\\
		DeepONet&1.22e-04& 6.20e-03 &   3.24e-03 & 0.50s &4.16M\\
		FNO&9.31e-05& 1.25e-03 &   1.26e-03 & 0.65s & 2.12M\\
		\hline
	\end{tabular}
	\label{tab:kpp-compare}
\end{table}

Our method demonstrates the best performance across all error metrics compared to FNO and DeepONet in Table \ref{tab:kpp-compare}.  Furthermore, our approach exhibits remarkable training efficiency, completing one epoch in only 0.03 seconds. This advantage is attributed to its intrinsic and straightforward architecture, which efficiently diagonalizes the linear system for time evolution.

\begin{figure}[!htbp]
	\centering
	\includegraphics[width=.5\textwidth]{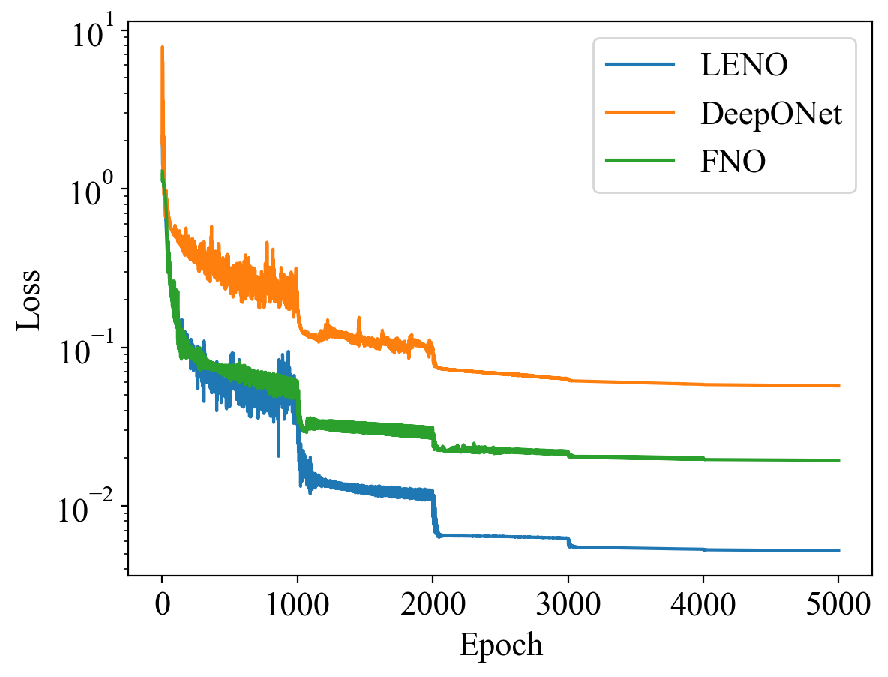}
	\caption{Loss curves of three different neural operators on the KPP-Fisher equation.} 
	\label{fig:kpp-loss}
\end{figure}

\paragraph{Comparison with Different Criteria}
We also evaluate the performance of the LE-NO trained with different loss criteria. Table \ref{tab:kpp-criterion} summarizes that training with both  $L^2$  loss ($L^D$) and residual loss ($L^R $) achieves the best performance across all metrics. In contrast, training solely with  $L^2$  loss results in poor learning of the nonlinear term, while incorporating residual loss  improves  $L^2$  error as well. These findings suggest that both  $L^2$  and residual errors are necessary in the loss function to ensure balanced and robust performance across all metrics.

\jdd{Furthermore, to validate Theorem~\ref{thm:opt}, we compare the training loss curves obtained using both the $L^2$ loss ($L^D$) and the residual loss ($L^R$) against those obtained using the residual loss alone. This experiment, conducted with a network of large width ($10000$), is illustrated in Figure~\ref{fig:losscompare}. As shown in the figure, the loss decreases slightly faster when both terms are included, which is consistent with the theoretical result that combining the two losses leads to a faster convergence rate.}

\begin{table}[!htbp]
	\centering
	\caption{Errors of different training criteria on the KPP-Fisher equation.}
	\begin{tabular}{c|ccc}
		\hline
		Criterion  &  $E_{L^2}$ &  $E_{Res}$  &$E_{Nonlinear}$  \\
		\hline
		\multirow{1}{*}{$L^D+L^R$}&5.79e-04&4.56e-03&4.57e-03\\
		
		\multirow{1}{*}{$L^R$} &1.03e-03&4.98e-03&4.98e-03\\
		
		\multirow{1}{*}{$L^D$}&4.96e-04&4.70e-01&4.70e-01\\
		\hline
	\end{tabular}
	
	\label{tab:kpp-criterion}
\end{table} 

\begin{figure}[!htbp]
    \centering
    \includegraphics[width=0.5\linewidth]{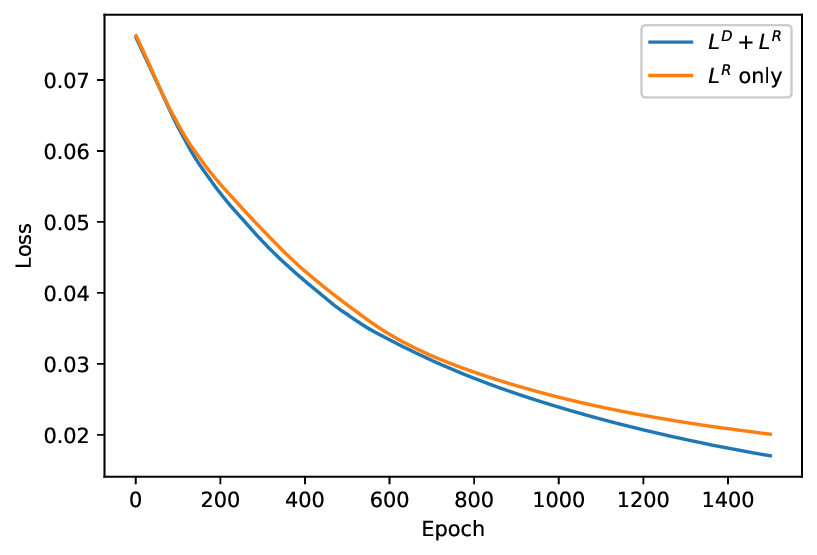}
    \caption{\jdd{Training loss curves under different criteria: combined loss $L^D + L^R$ versus residual loss $L^R$ only.}}
    \label{fig:losscompare}
\end{figure}

\paragraph{Performance versus Various $P$} \jd{Note that the DeepONet is structured such that, for given sensor points $x_j$, $j=1,2,\dots,m$, it employs $P$ branch networks $b_k$ and $P$ trunk networks $t_k$. Its output is defined as
$$
\text{DeepONet}(u)(y) = \sum_{k=1}^P b_k(u(x_1),u(x_2),\cdots,u(x_m))t_k(y)
$$}
\jd{We evaluate the performance of our method using different numbers of eigenfunctions $P$, while setting the number of branch and trunk networks equal to $P$. The results, obtained using networks with a slightly smaller number of neurons, are summarized in Table~\ref{tab:kpp-variousp}.} DeepONet consistently performs worse across all metrics compared to our method. The results indicate that increasing  $P$  generally reduces the error for both methods. However, when  $P$  becomes very large, the error reduction slows, as it becomes influenced by other factors, such as optimization and discretization errors. To achieve a balance between computational efficiency and accuracy, we recommend using a moderately large  $P$. 

\begin{table}[!htbp]
	\centering
	\caption{Errors versus various $P$.}
	\begin{tabular}{c| c c|cc|cc}
		\hline
		\multirow{2}{*}{$P$} & \multicolumn{2}{c}{$E_{L^2}$} & \multicolumn{2}{c}{$E_{Res}$} & \multicolumn{2}{c}{$E_{Nonlinear}$ }\\
		\cline{2-7}
		& Our & DeepONet& Our & DeepONet& Our & DeepONet\\
		\hline
		2&2.25e-02 &1.18e-02&2.09e-02 &6.55e-02&2.02e-01&2.66e-01\\
		4&1.81e-03&8.17e-03&4.03e-03&7.33e-02&4.79e-02&1.92e-01\\
		8& 5.62e-04&7.47e-03&3.89e-03&7.28e-02& 9.17e-03&1.86e-01\\
		16&5.61e-04&6.84e-03&4.20e-03&5.76e-02&4.62e-03&1.83e-01\\
		32&5.81e-04&6.18e-03&4.42e-03&5.08e-02& 4.44e-03&1.83e-01\\
		64& 5.79e-04&6.23e-03&4.56e-03&5.08e-02&4.57e-03&1.83e-01\\
		\hline
	\end{tabular}
	
	\label{tab:kpp-variousp}
\end{table}

\jdall{We also investigate the convergence behavior of the $L^2$ solution error $E_{L^2}$ and the $L^2$ nonlinear term error $E_R$ with respect to the number of eigenfunctions $P$, as illustrated in Figure~\ref{fig:errvaryP}. This study serves both to validate the theoretical results established in Theorems~\ref{thm:approxop} and~\ref{thm:err-sol} and to provide an ablation analysis on the role of $P$. Since the learned neural operator does not provide the exact solution, optimization error is inevitably present, and moreover, it is generally difficult to quantify the approximation error of neural networks. Therefore, we focus our validation on the dependence with respect to $P$. For small values of $P$, both the optimization error and the neural network approximation error remain negligible. In this case, both the solution and nonlinear term errors exhibit a convergence rate of order $2$, which is consistent with theory: the initial condition is sampled from $\mu=\mathcal{N}(0,49(-\Delta+7I)^{-2.5})$, which possesses $H^{2-\delta}$ regularity for any $\delta>0$. Consequently, the theoretical rate corresponds to $k=2-\delta$ in one spatial dimension, in agreement with the numerical results. For larger values of $P$, the optimization error dominates and the convergence rate saturates, indicating that excessively large choices of $P$ are unnecessary. Rather, an appropriate balance between discretization error and optimization error should be sought.
}

\begin{figure}[!htbp]
    \centering
    \includegraphics[width=0.5\linewidth]{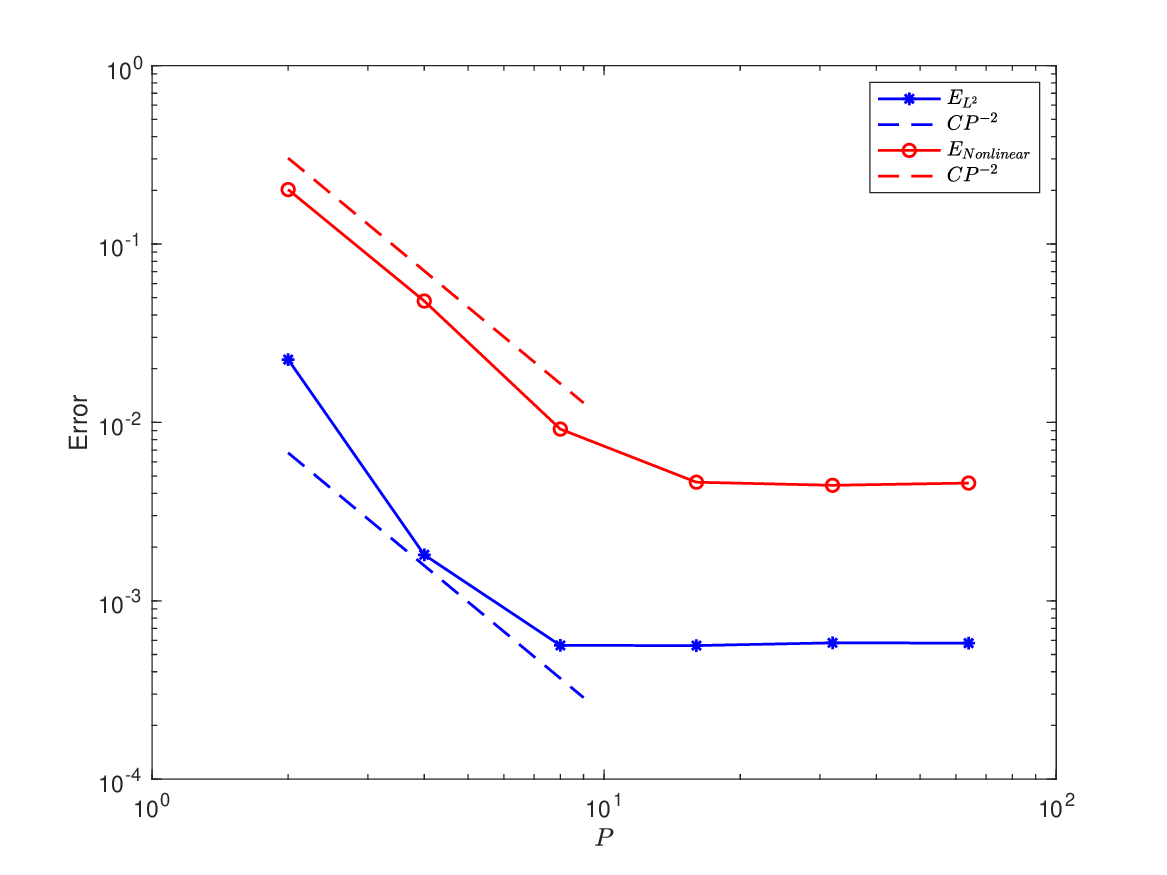}
    \caption{\jdd{$L^2$ solution error $E_{L^2}$ and $L^2$ nonlinear term error $E_R$ as functions of the number of eigenfunctions $P$.}}
    \label{fig:errvaryP}
\end{figure}

\jddd{\paragraph{Ablation Study on Neural Network Size}  
We also conduct an ablation study on the neural network architecture by varying the width $w$ and depth $l$ of the network to evaluate their effect on the $L^2$ errors of both the solution and the nonlinear term, as illustrated in Figure~\ref{fig:errvaryNN}. The results indicate that deeper and wider networks generally lead to lower errors. However, for sufficiently deep networks, increasing the width further yields diminishing improvements, with the error reduction rate becoming slower. These observations suggest that a depth of $2$–$3$ and a width of approximately $1000$ provide a good balance between accuracy and efficiency. }

\begin{figure}[!htbp]
    \centering
    \includegraphics[width=0.45\linewidth]{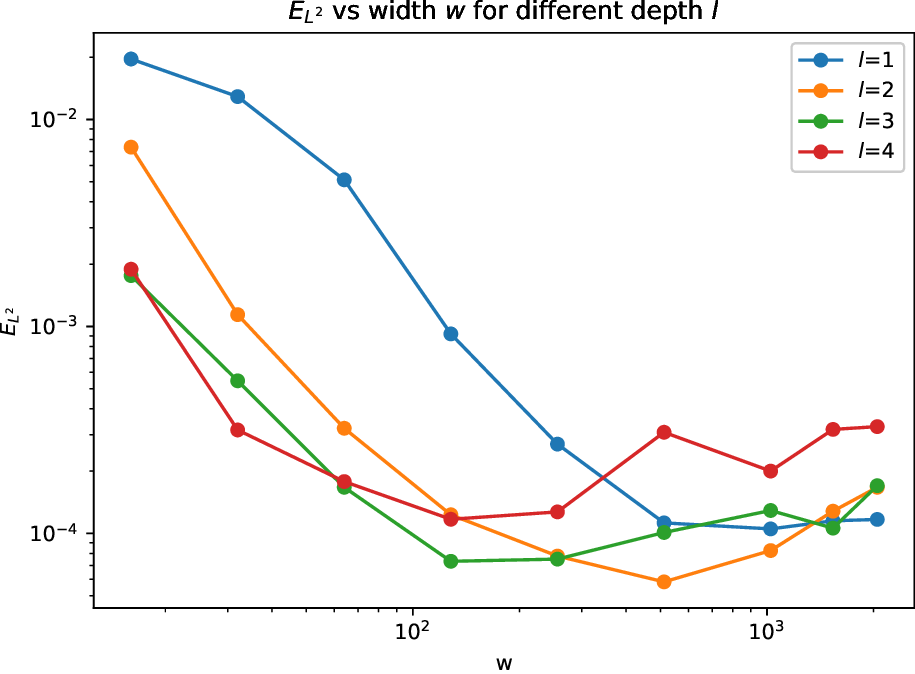}
    \includegraphics[width=0.45\linewidth]{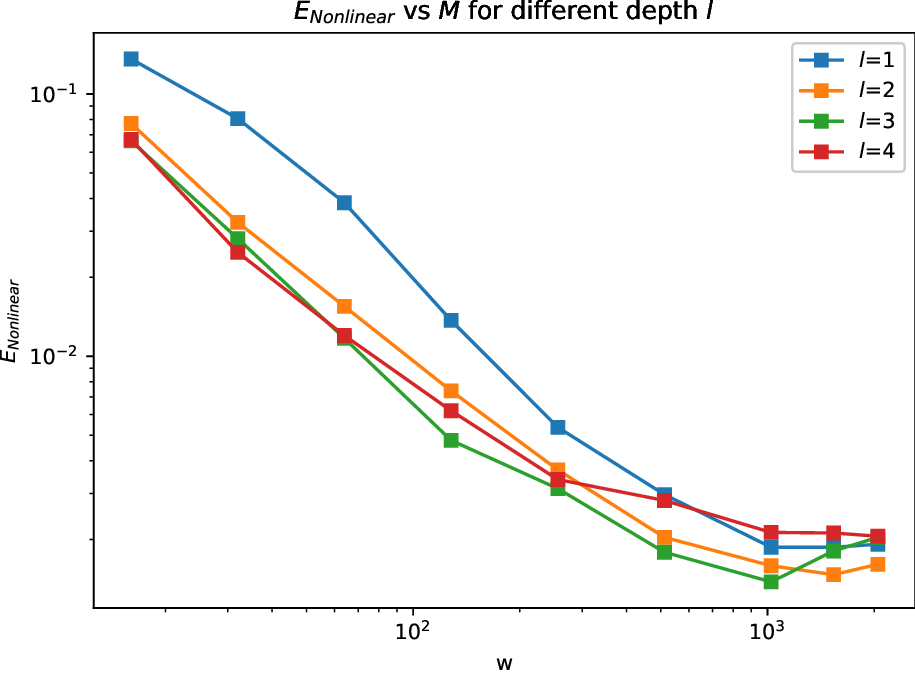}
    \caption{\jddd{$L^2$ errors of the solution and nonlinear term for different network widths and depths.}}
    \label{fig:errvaryNN}
\end{figure}

\paragraph{Dynamics Prediction Results}
We also use our method to predict the solution after training. As shown in Figure \ref{fig:kpp-predict}, we present the  $L^2$-norm of the solution with respect to time $t$ . The training process covers the first 10 time steps, represented by the shaded area, while the prediction extends beyond this region. The results demonstrate that our method provides accurate predictions, maintaining consistency with the trained region and effectively predicting to the unknown time steps.

\begin{figure}[!htbp]
	\centering
	\includegraphics[width=.45\textwidth]{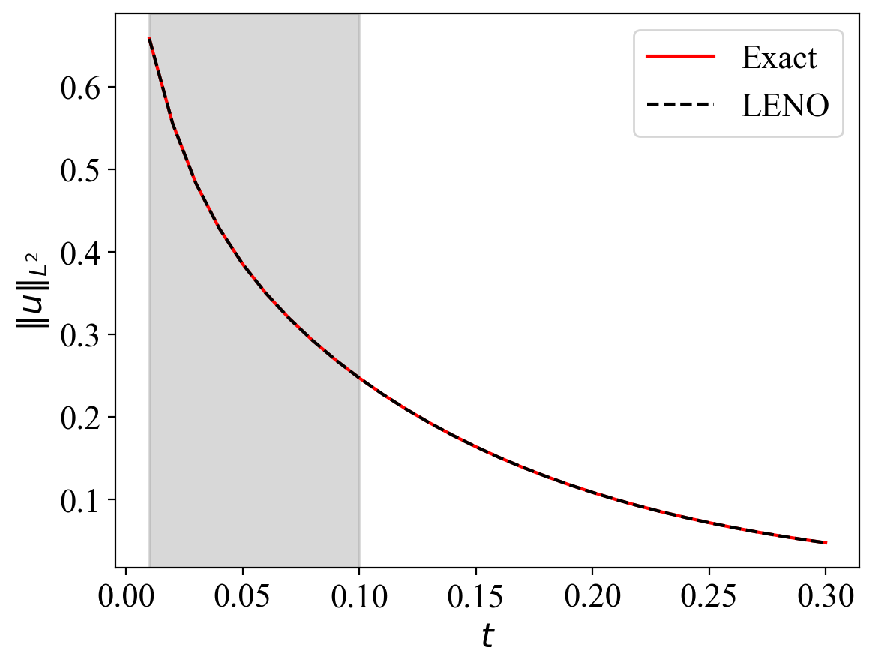}
	\includegraphics[width=.45\textwidth]{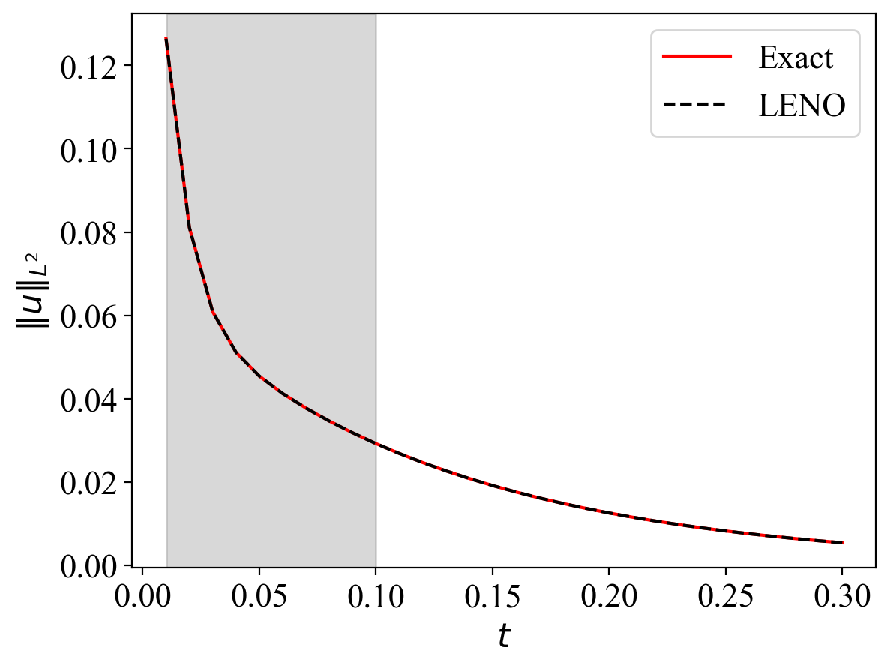}
	\caption{Predicted results ($L^2$-norm) for two initial conditions on the KPP-Fisher equation. The shaded area represents the time steps used for training.}
	\label{fig:kpp-predict}
\end{figure}

\paragraph{Inhomogeneous Boundary Condition Case}
We also test our method with an inhomogeneous boundary condition:
\begin{equation}
	\begin{aligned}
		u(0,t)&=u(1,t) = 1, \quad &t\in (0,T),\\
	\end{aligned}
\end{equation}
Our method achieves a performance similar to that observed in the homogeneous setting, namely 
$E_{L^2}=5.56e-04$, $E_{Res}=5.27e-03$, and $E_{Nonlinear}=5.29e-03$. The results demonstrate that our method is effective and adaptable to problems with inhomogeneous boundary conditions. Additionally, Figure \ref{fig:kppinh-predict} illustrates that our method delivers accurate predictions of the solution, as evidenced by the  $L^2$-norm evolution over time.  The method generalizes effectively, providing accurate predictions beyond the initial 10 steps of training.

\begin{figure}[!htbp]
	\centering
	\includegraphics[width=.45\textwidth]{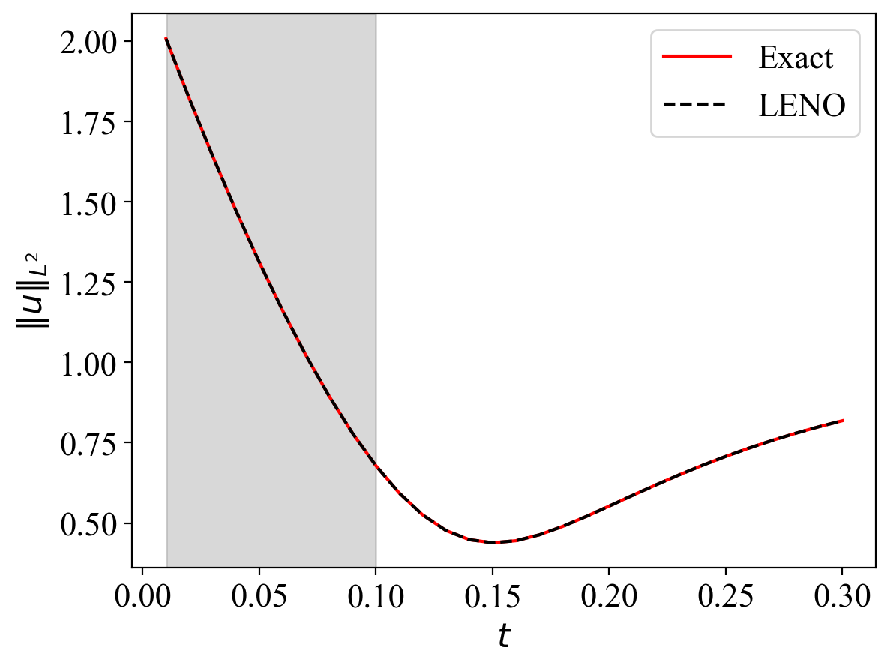}
	\includegraphics[width=.45\textwidth]{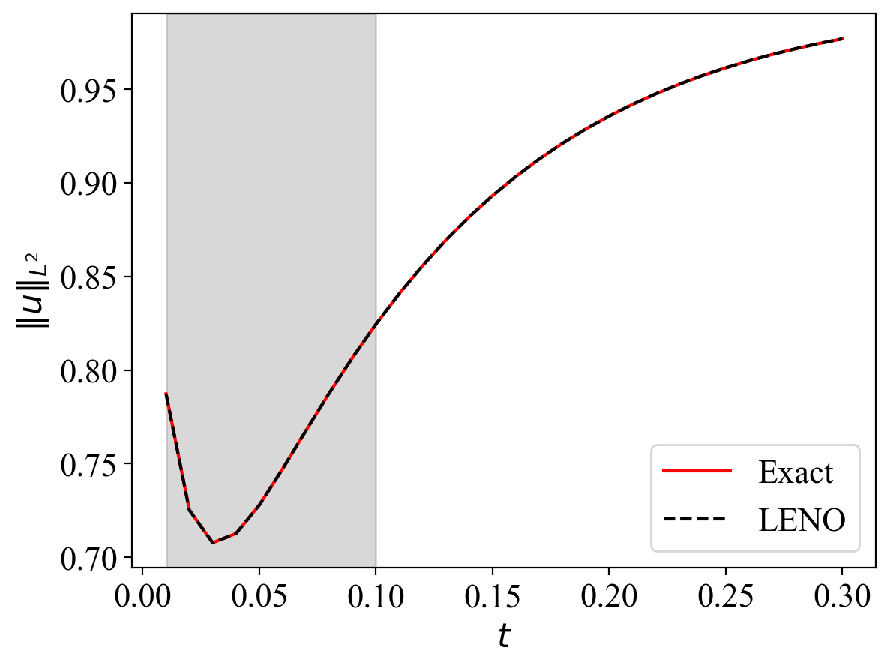}
	\caption{Predicted results ($L^2$-norm) for two initial conditions on the KPP-Fisher equation with inhomogeneous boundary condition. The shaded area represents the time steps used for training.}
	\label{fig:kppinh-predict}
\end{figure}

\jdall{
\subsubsection{Extension to Variable and Anisotropic Diffusion Coefficients}  

In this section, we validate the capability of our method in handling variable and matrix-valued diffusion coefficients, including anisotropic cases. We consider two representative examples based on the KPP–Fisher equation:  

\begin{itemize}
    \item \textbf{Case 1.} 1D KPP–Fisher equation with a variable diffusion coefficient:  
    \begin{equation}
        \frac{\partial u}{\partial t} - \frac{\partial}{\partial x}\!\left((2+\cos(\pi x))\frac{\partial u}{\partial x}\right) = u(1-u), 
        \quad x \in (0,1), \; t \in (0,T).
    \end{equation}

    \item \textbf{Case 2.} 2D KPP–Fisher equation with an anisotropic diffusion coefficient:  
    \begin{equation}
        \frac{\partial u}{\partial t} - \nabla \cdot (D \nabla u) = u(1-u), 
        \quad x \in (0,1)^2, \; t \in (0,T),
    \end{equation}
    where 
    \begin{equation}
        D = 
        \begin{pmatrix}
            1 & 0 \\
            0 & 0.001
        \end{pmatrix}.
    \end{equation}
\end{itemize}

In both cases, we impose homogeneous Dirichlet boundary conditions. The numerical errors are summarized in Table~\ref{tab:errvaryD}. The method achieves low $L^2$ errors for both the solution and the nonlinear term, consistent with the constant-coefficient setting. This indicates that the eigenfunction basis adapted to the variable diffusion coefficient continues to provide an effective representation. Furthermore, Figures~\ref{fig:varyDcase1} and~\ref{fig:varyDcase2} present the prediction results for the two cases, demonstrating that the method yields accurate predictions even for long time steps. In particular, we show the ground-truth solution, the prediction of our method at $t=0.3$ (after training up to $t=0.1$), and the corresponding absolute error in Figure~\ref{fig:kpperrcompare}. The results indicate that even in the presence of highly anisotropic diffusion, the eigenfunction basis still provides accurate approximations and strong learning capability, leading to reliable predictive performance.  

}

\begin{table}[!htbp]
\centering
\jdall{
    \caption{\jdall{Errors of the KPP-Fisher equations with different diffusion coefficients.} }
    \label{tab:errvaryD}
	\begin{tabular}{c|ccc}
		\hline
		Diffusion &  $E_{L^2}$ &  $E_{Res}$  &$E_{Nonlinear}$ \\
		\hline
		Case 1 &  1.27e-04 & 2.79e-03 & 2.79e-03\\
		Case 2 & 3.62e-03 & 3.70e-03 & 1.72e-02\\
		\hline
	\end{tabular}}
\end{table}

\begin{figure}[!htbp]
    \centering
    \includegraphics[width=0.45\linewidth]{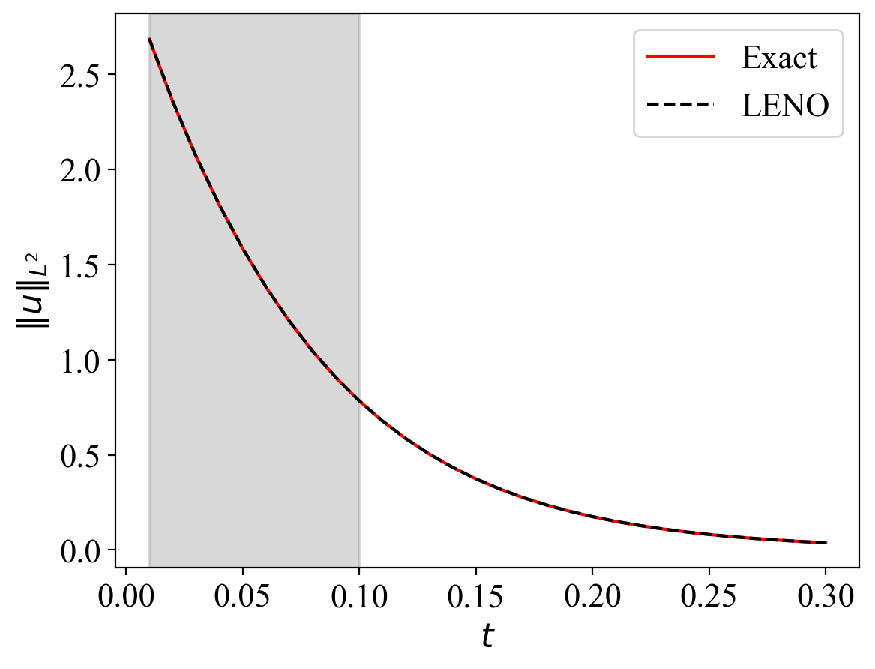}
    \includegraphics[width=0.45\linewidth]{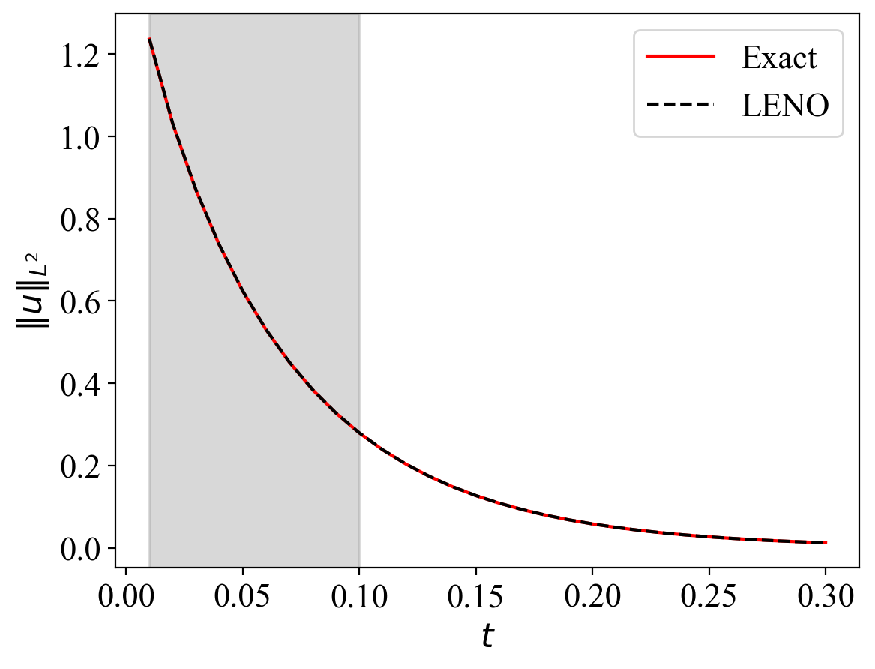}

    \caption{\jdall{Case 1. Predicted results ($L^2$-norm) for two initial conditions on the KPP-Fisher equation with variable diffusion coefficient. The shaded area represents the time steps used for training.}}
    \label{fig:varyDcase1}
\end{figure}

\begin{figure}[!htbp]
    \centering
    \includegraphics[width=0.45\linewidth]{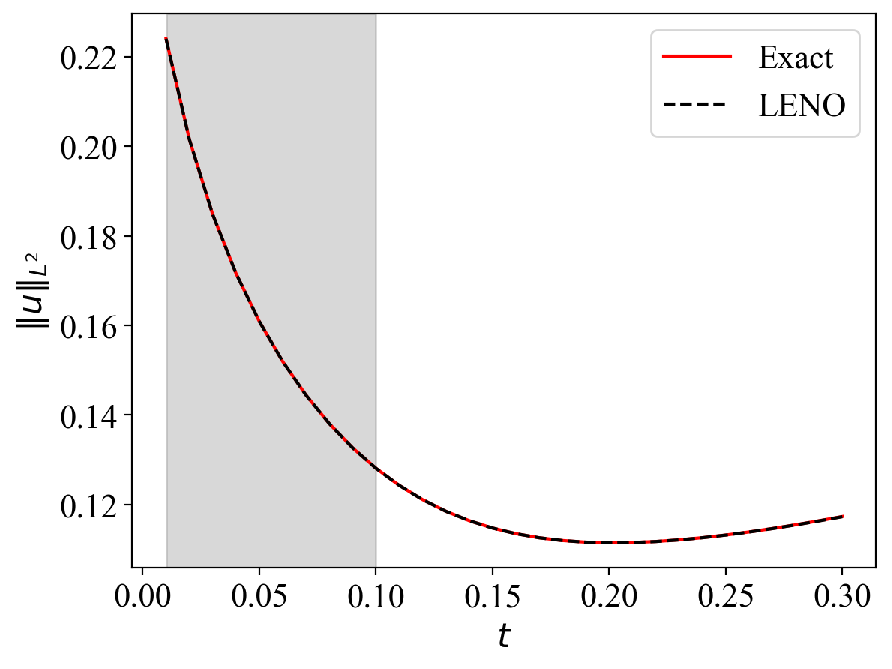}
    \includegraphics[width=0.45\linewidth]{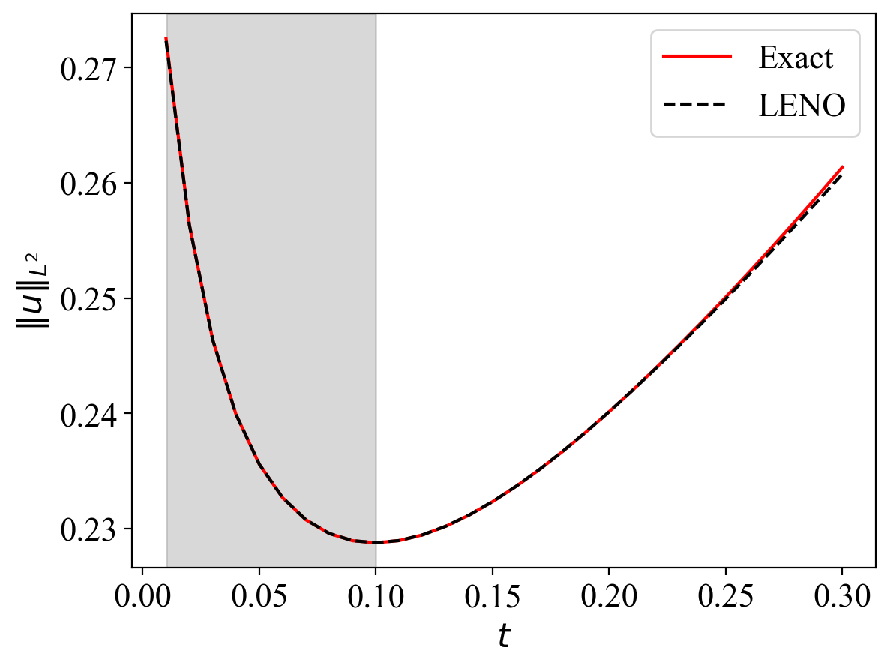}

    \caption{\jdall{Case 2. Predicted results ($L^2$-norm) for two initial conditions on the KPP-Fisher equation with anisotropic diffusion coefficient. The shaded area represents the time steps used for training.}}
    \label{fig:varyDcase2}
\end{figure}

\begin{figure}[!htbp]
    \centering
    \includegraphics[width=0.9\linewidth]{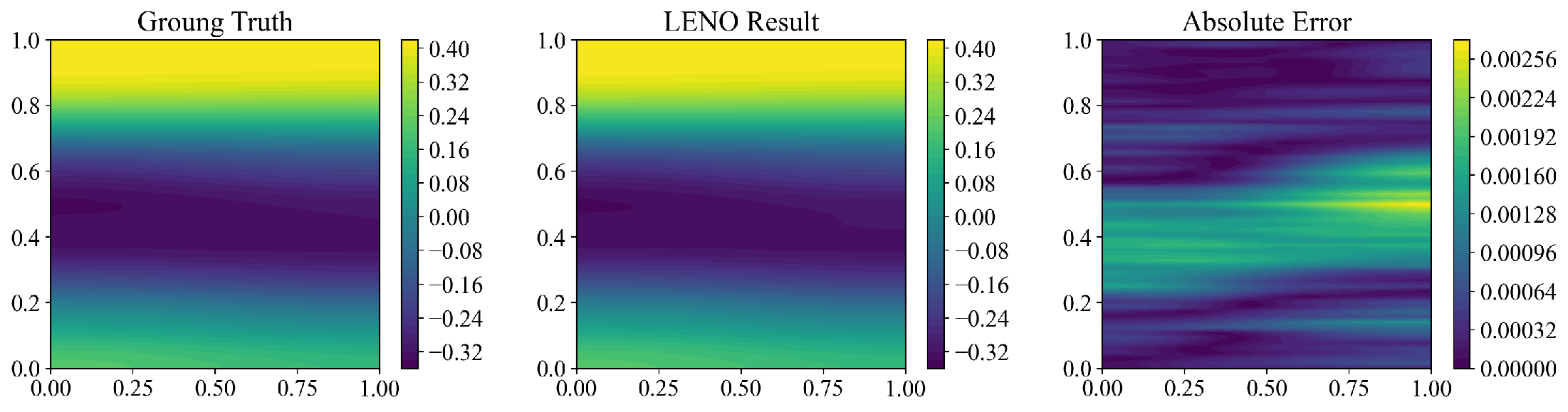}
    \caption{\jdall{Case 2. Comparison of the ground truth (left), LENO prediction (middle), and absolute error (right) for the KPP–Fisher equation with an anisotropic diffusion coefficient.}}
    \label{fig:kpperrcompare}
\end{figure}

\subsection{Example 2: Allen-Cahn Equation}\label{subsec:ac}
We consider the 1D Allen-Cahn equation  which describes the process of phase separation in multi-component alloy systems with a non-flux boundary condition
\begin{equation}
	\left\{	\begin{aligned}
		{\partial u\over \partial t} - \varepsilon^2 {\partial^2 u\over\partial x^2} &= W'(u),  &x\in (0,2\pi),t\in (0,T]\\
		u(x,0) &= u_0(x),&t\in (0,T)\\
	\end{aligned}\right.
\end{equation}
where $W(u)={1\over 4}(u^2-1)^2$ and we set a small diffusion coefficient $\varepsilon=0.1$ in this test.

The results of our experiments are presented in Figure \ref{fig:ac}. In this small diffusion case, our proposed method achieves a similar level of accuracy as in the KPP-Fisher equation. Additionally, we compare the predicted results of our method to those of FNO and DeepONet. We also compare the predicted results of our method to FNO and DeepONet. As shown in Figure \ref{fig:ac}, DeepONet fails to produce satisfactory results even during the time steps used for training. While FNO shows a good fit during the training steps, its predictions diverge significantly from the true solution within just a few time steps. In contrast, our method maintains reliable prediction accuracy over an extended period, showcasing its ability to generalize and provide stable predictions.

\begin{figure}[!htbp]
	\centering
	\includegraphics[width=.45\textwidth]{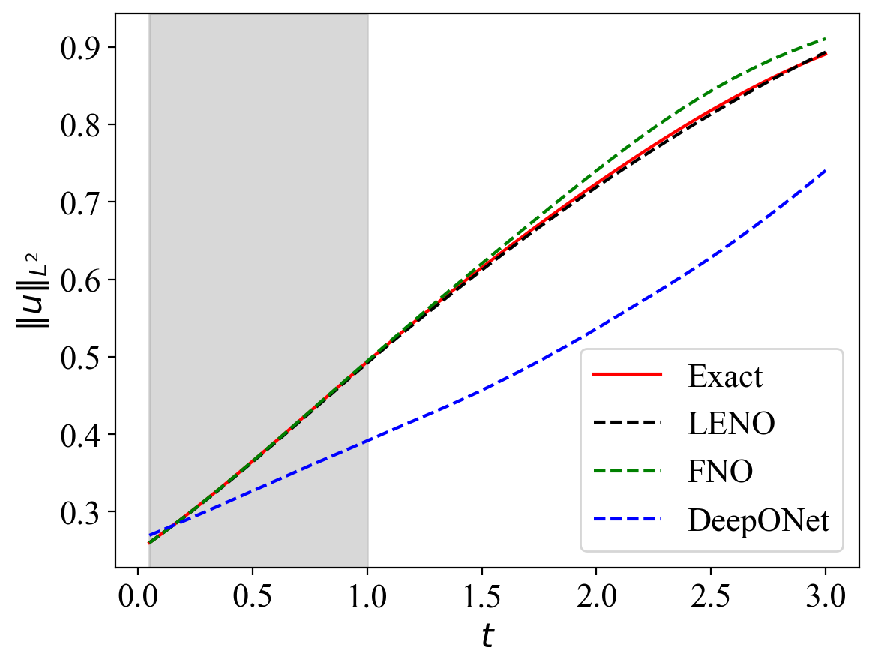}
	\includegraphics[width=.45\textwidth]{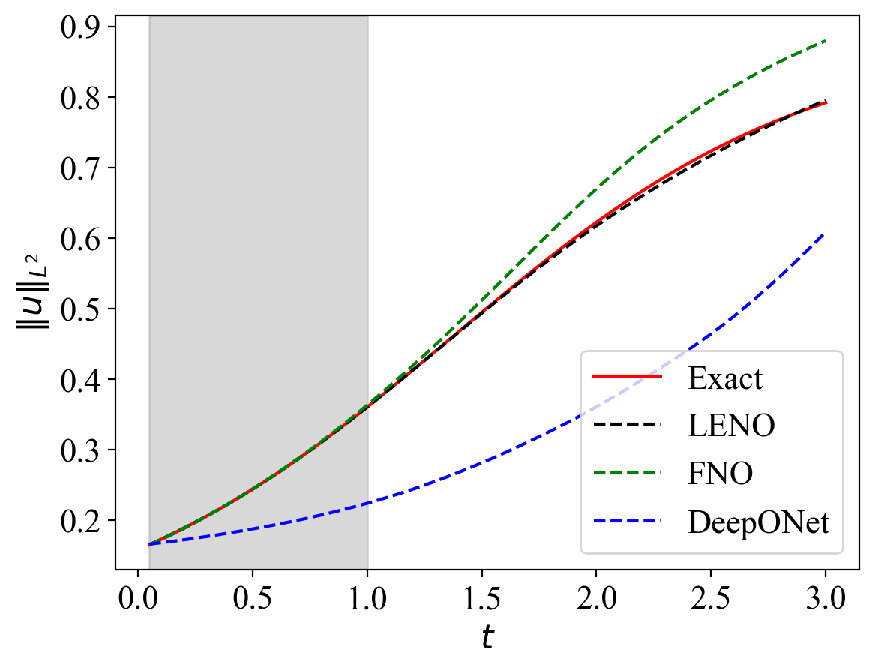}
	\caption{Predicted results ($L^2$-norm) for two initial conditions  on the Allen-Cahn equation. The shaded area represents the time steps used for training.}
	\label{fig:ac}
\end{figure}

\subsection{Example 3: Gray-Scott Equations}\label{subsec:gs}
We consider the following Gray-Scott equations which is a reaction--diffusion system in 1D ($d=1$) and 2D ($d=2$) with non-flux boundary conditions:
\begin{equation}
	\left\{	\begin{aligned}
		A_t-D_A\Delta A &= SA^2-(\mu+\rho)A, &x \in \Omega,t\in(0,T)\\
		S_t-D_S\Delta S &= -SA^2+\rho(1-S), &x\in \Omega,t\in(0,T)\\
		A(x,0) &= A_0(x)  &x \in \Omega\\
		S(x,0) &= S_0(x)  &x \in \Omega\\
	\end{aligned}\right.
\end{equation}
where we set $\Omega=(0,2\pi)^d$, $D_A=2.5\times 10^{-4}$, $D_S=5\times10^{-4}$, $\rho=0.04$ and $\mu=0.065$.

\paragraph{1D Results}

The quantitative results of error in Table \ref{tab:gs1d} confirm the effectiveness of our approach. Additionally, the predicted results in Figure \ref{fig:gs1d} demonstrate that our method still performs well for a system.

\begin{table}[!htbp]
	\centering
	\caption{Errors of the Gray-Scott equations. }
	\begin{tabular}{c|ccc}
		\hline
		Variable &  $E_{L^2}$ &  $E_{Res}$  &$E_{Nonlinear}$ \\
		\hline
		\multicolumn{4}{c}{1D Problem}\\
		\hline
		$A$ &  9.55e-04 & 7.20e-03 & 8.04e-03\\
		$S$ & 9.46e-04& 6.50e-03 & 6.34e-03\\
		\hline
		\multicolumn{4}{c}{2D Problem}\\
		\hline
		$A$ &  5.31e-03& 2.32e-02  & 2.71e-02\\
		$S$ & 4.13e-03 & 1.79e-02 & 1.66e-02\\
		\hline
	\end{tabular}
	\label{tab:gs1d}
\end{table}

\begin{figure}[!htbp]
	\centering
	\includegraphics[width=.75\textwidth]{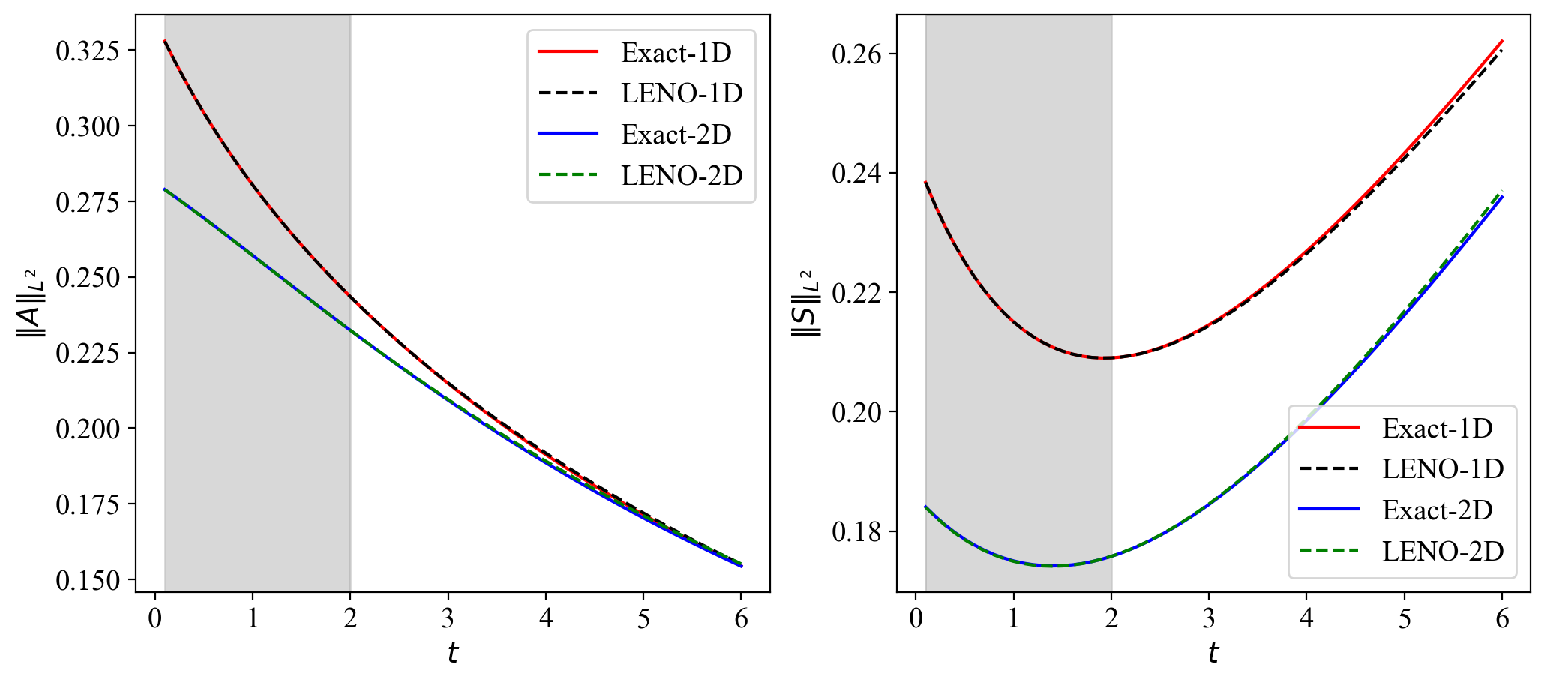}
	\caption{Predicted results ($L^2$-norm) of $A$ (left) and $S$ (right) on the 1D and 2D  Gray-Scott equations. The shaded area represents the first 20 time steps used for training.}
	\label{fig:gs1d}
\end{figure}

\paragraph{2D Results}

For the 2D Gray-Scott system, our method achieves slightly lower accuracy compared to the 1D case, as shown in Table \ref{tab:gs1d}. However, the performance remains satisfactory, demonstrating the capability of our approach to handle more complex spatial dimensions effectively. The predicted results, presented in Figure \ref{fig:gs1d}, further confirm the reliability of our method. Similar to the 1D case, the model provides accurate predictions beyond the training phase, as evidenced by the  $L^2$ -norm evolution for both variables  $A$  and  $S$ . The shaded regions in the plots represent the first 20 time steps used for training, during which the method captures the system dynamics effectively. Beyond these steps, the predictions remain consistent and closely align with the true solutions.

\subsection{Example 4: Schr\"odinger Equation}\label{subsec:sch}

We consider the following Schr\"odinger equation with a homogeneous Dirichlet boundary condition:
\begin{equation}
	\left\{\begin{aligned}
		&u_t	-\Delta u +\alpha |u|^2u+Vu=\lambda u, &x\in \Omega,t\in(0,T)\\
		& u(x,0)=u_0(x) &x\in \Omega\\
	\end{aligned}
	\right.
\end{equation}
where we set $\Omega=(-8,8)^2$, $\alpha=1600$, $\lambda=15.87$ is the eigenvalue of the steady problem and the potential is given by:
\begin{equation}
	V(x,y) = 100(\sin^2({\pi x\over 4})+\sin^2({\pi y\over 4}))+x^2+y^2.
\end{equation}
This equation features a more complex nonlinear term, which adds to the challenge of accurately solving and predicting the system dynamics.  As summarized in Table \ref{tab:Sch}, the results demonstrate that our method achieves low errors, although the difference between the residual error and nonlinear error is slightly larger, reflecting the increased complexity of the nonlinear interactions. Furthermore, Figure \ref{fig:Sch} illustrates the predicted  $L^2$-norm evolution over time for the 2D Schr\"odinger equation. As time progresses, the predicted results show a slight deviation from the actual solution, primarily due to the accumulation of errors over time. This effect becomes more noticeable for such a complex problem. Nevertheless, within the same range of steps as the training phase, our method continues to deliver a good performance.

\begin{figure}[!htbp]
	\centering
	\includegraphics[width=.45\textwidth]{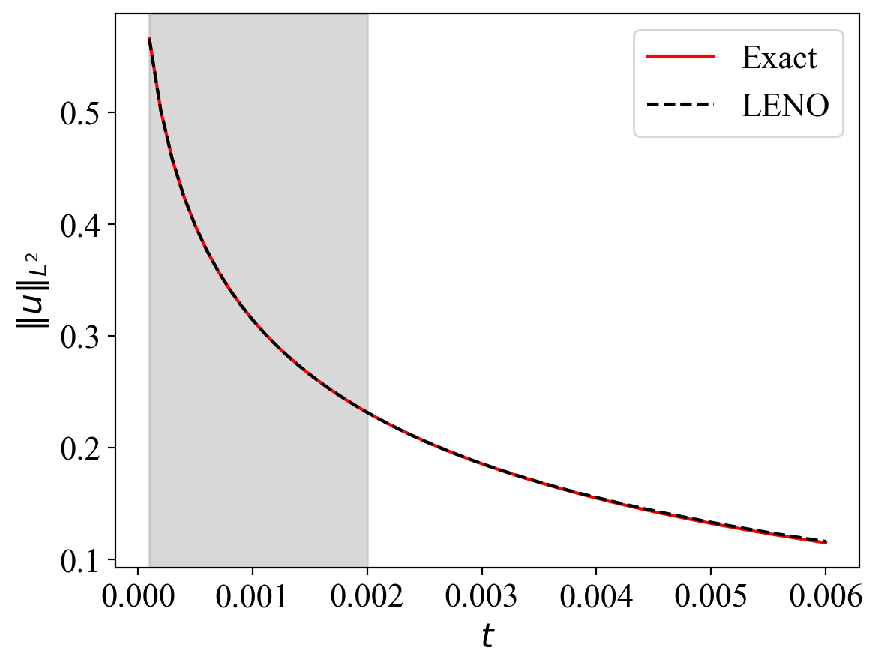}
	\includegraphics[width=.45\textwidth]{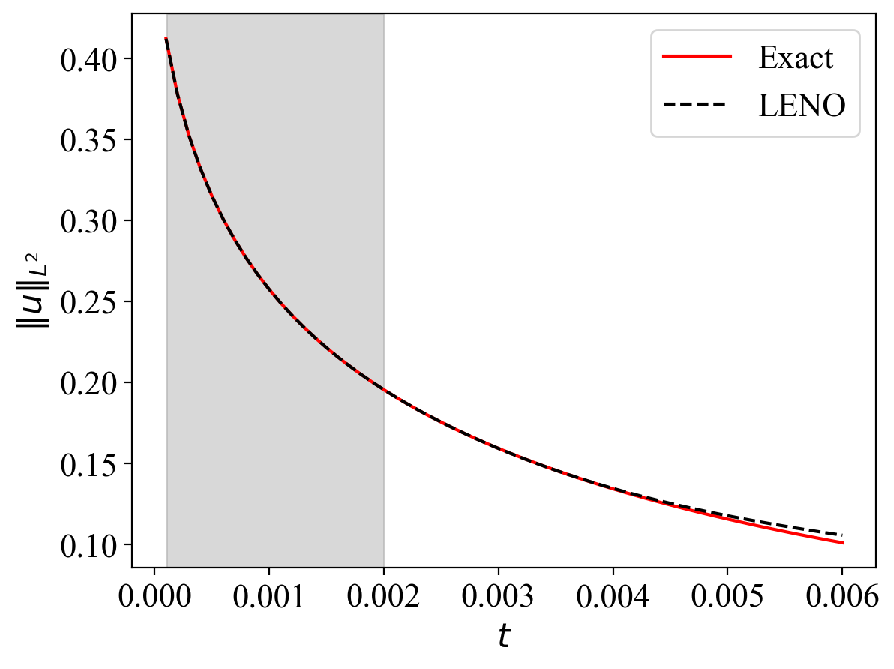}
	\caption{Predicted results ($L^2$-norm) for two initial conditions on the 2D Schr\"odinger equation. The shaded area represents the first 20 time steps used for training.}
	\label{fig:Sch}
\end{figure}

\subsection{Example 5: Alzheimer's Disease Evolution}\label{subsec:brain}

We analyze the dynamics of cortical thickness in MRI scan slices of Alzheimer's disease patients using the Alzheimer's Disease Neuroimaging Initiative (ADNI) dataset, governed by the following equation:	\begin{equation}
	\left\{	\begin{aligned}
		u_t - D \Delta u & = \mathcal{F}(u),& x\in \Omega\subset\mathbb{R}^2,t\in (0,T)\\
		u(x,0) &= u_0(x) ,& x\in \Omega
	\end{aligned}\right.\label{AD}
\end{equation}
with a pure Neumann boundary condition. The diffusion coefficient $D$ is unknown and is regarded as a  parameter in the training. 
Figure \ref{fig:ctscans} shows MRI scans of a human brain at three different time points (Baseline,  6 months, and  12 months) for a specific patient. These slices highlight subtle structural changes in the brain over time, potentially indicating the presence of disease or other pathological conditions. By modeling this evolution, we aim to accurately capture these changes through learning the nonlinear term $\mathcal{F}(u)$, providing deeper insights into potential abnormalities.

\begin{figure}[!htbp]
	\centering
	\subfloat[Baseline]{\includegraphics[width=.32\textwidth]{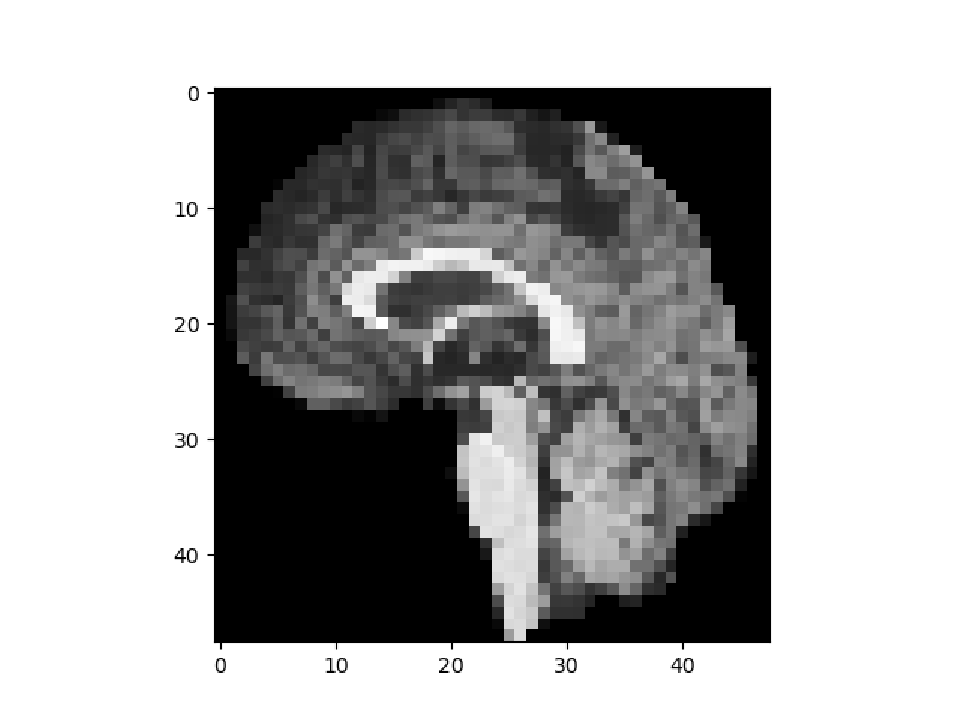}}
	\subfloat[6 months]{\includegraphics[width=.32\textwidth]{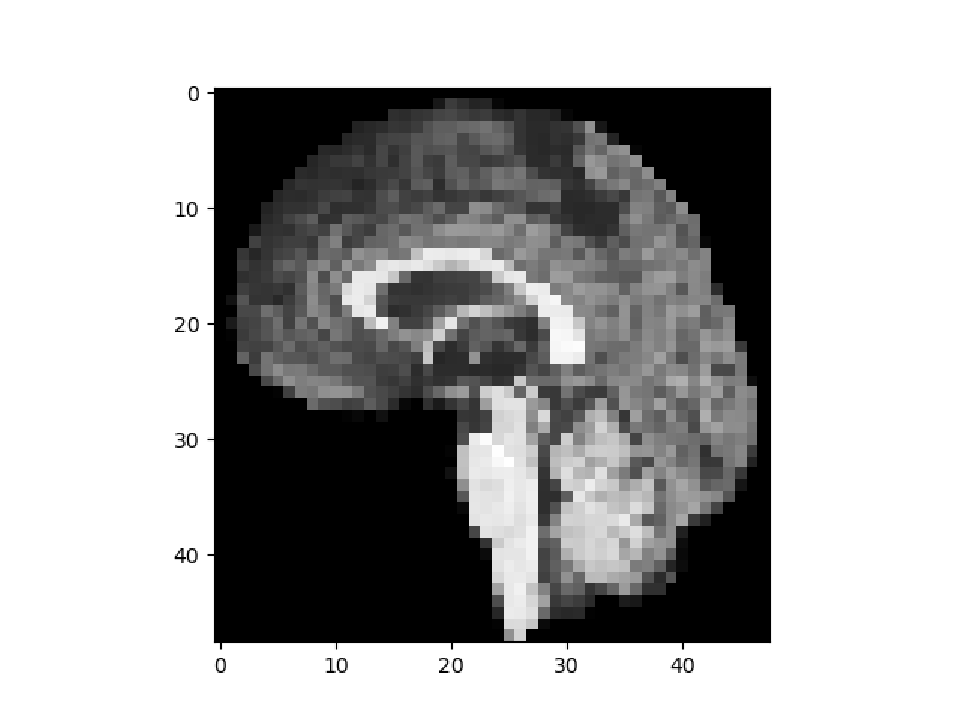}}
	\subfloat[12 months]{\includegraphics[width=.32\textwidth]{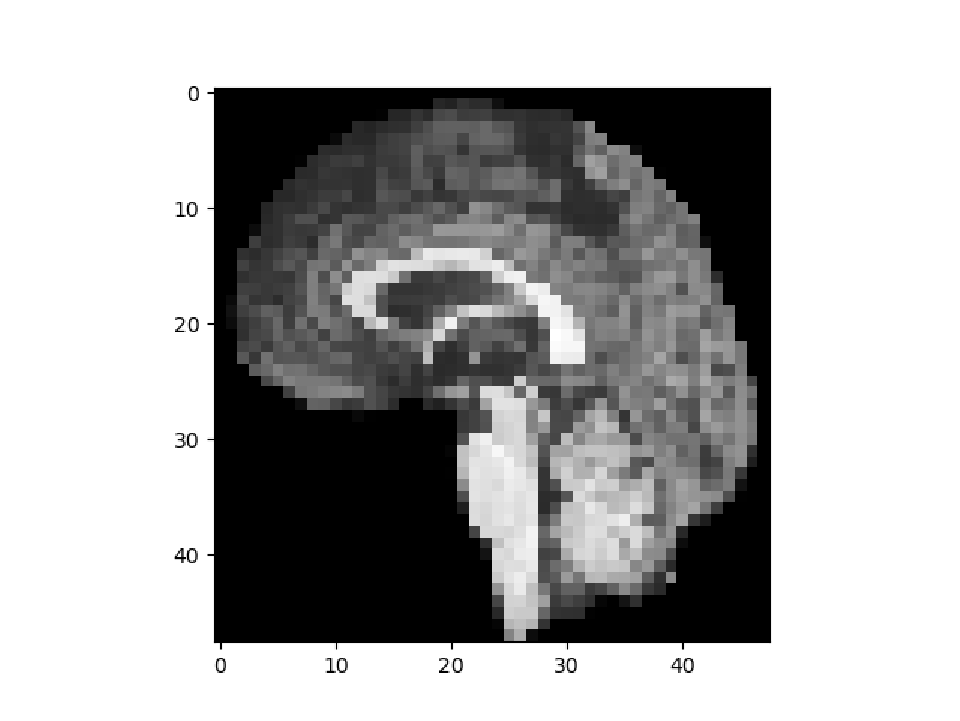}}
	\caption{MRI slices of an Alzheimer's disease patient at three different time points: baseline, 6 months, and 12 months.}
	\label{fig:ctscans}
\end{figure}

In this study, we consider an irregular domain $\Omega$ representing the human brain and use the FEniCS eigensolver to compute the corresponding Laplacian eigenfunctions. These eigenfunctions serve as the basis for our framework, which is employed for both training and transfer learning to model the progression of Alzheimer’s disease. As illustrated in Fig. \ref{fig:eigenfun}, the eigenfunctions localize to different spatial regions of the brain geometry, providing a natural source of interpretability for the LE-NO framework.

\begin{figure}[!htbp]
	\centering
	\includegraphics[width=0.32\linewidth]{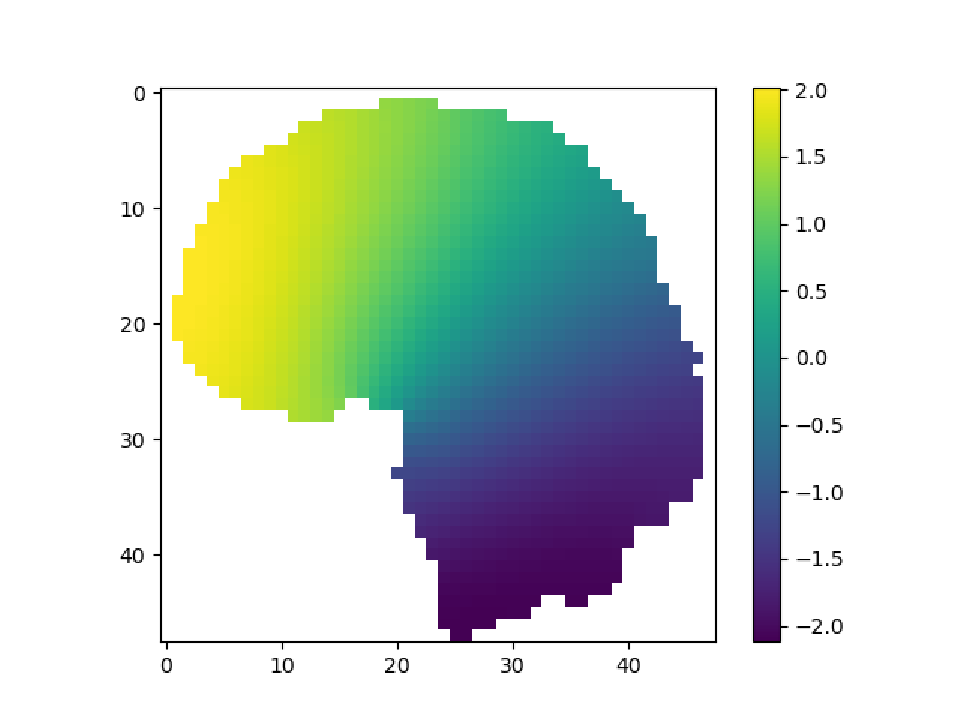}
	\includegraphics[width=0.32\linewidth]{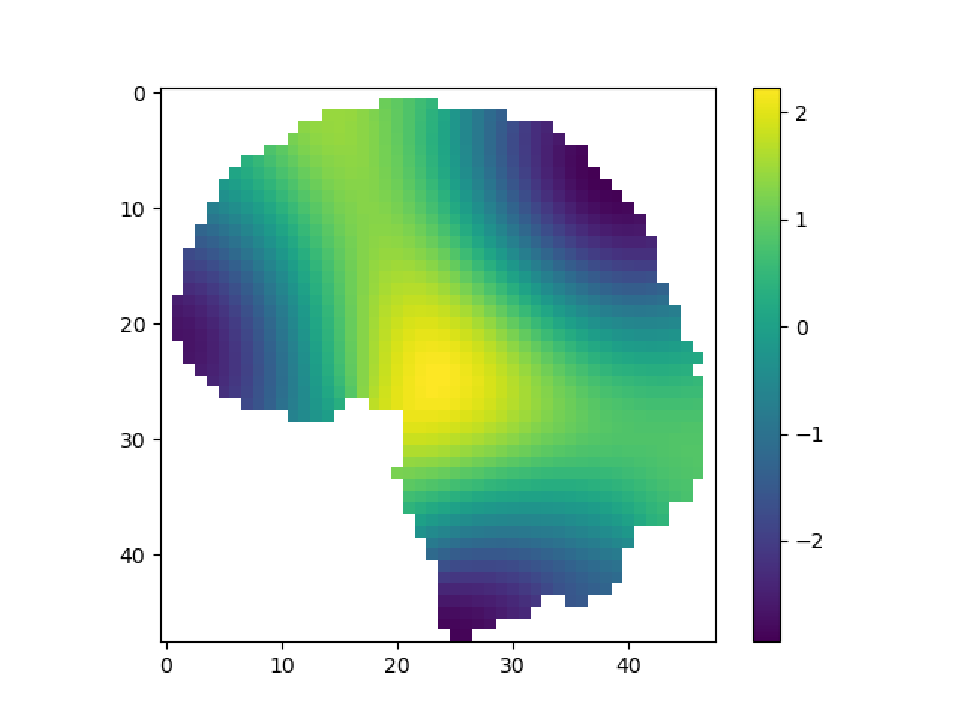}
	\includegraphics[width=0.32\linewidth]{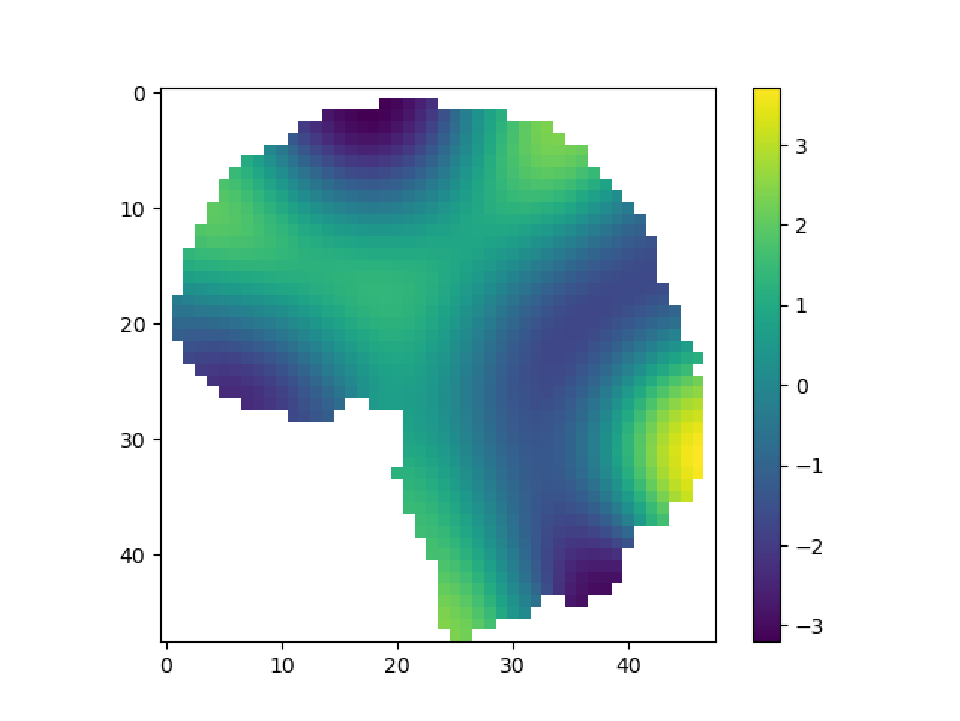}
	\caption{Different Laplacian eigenfunctions computed on an irregular domain representing the human brain.}
	\label{fig:eigenfun}
\end{figure}

Table \ref{tab:ct} presents the accuracy results for learning the eigen-coefficients associated with $L^D$. The results show that our method achieves an accuracy of at least 90\% across all metrics, highlighting its effectiveness in capturing the dynamic behavior. Figure \ref{fig:cttransfer} offers a visual comparison between the predicted results and the ground truth, demonstrating a close match.

\paragraph{Transfer Learning}	
We apply transfer learning across different patients. In this approach, the nonlinear layers of the neural network are kept fixed, with only the linear regression layer being retrained. This allows for efficient adaptation of the model to new spatial regions with minimal training effort. When transferring to a different patient, the time scale may vary due to the heterogeneity among patients. To address this, we introduce a disease progression score $s = \alpha t + \beta$, where $\alpha$ and $\beta$ are constants \cite{ghazi2021robust,zheng2022data}. The model in (\ref{AD}) becomes:
$\alpha u_t - D \Delta u = \mathcal{N}(u)$ In this case, we use transfer learning to adapt the model for different patients by training the additional parameter $\alpha$.

Table~\ref{tab:ct} also summarizes the transfer learning accuracy for different patients. The results demonstrate consistent accuracy ($\ge90\%$) with previous tests. Figure~\ref{fig:cttransfer} provides a visual comparison of the LE-NO training and prediction results compared with ground truth, showing that the outcomes closely align with the true dynamics.

\begin{table}[!htbp]
	\centering
	\caption{Training and transfer learning accuracy results for irregular domain.}
	\begin{tabular}{c|ccc}
		\hline
		\multicolumn{4}{c}{Learning results}\\
		\hline
		& $1-L^D$ ( \%)& $1-E_{L^2}$  (\%) & $1-E_{Res}$ (\%)\\
		\cline{2-4}
		& 94.82 & 91.01 & 96.67\\
		\hline
		\multicolumn{4}{c}{Transfer learning results}\\
		\hline
		Patient & $1-L^D$ ( \%)& $1-E_{L^2}$  (\%) & $1-E_{Res}$ (\%)\\
		\hline
		1 & 97.68 & 91.42 & 90.23 \\
		2 & 95.94 & 90.53 & 93.73 \\
		3 & 96.69 & 90.36 & 86.83 \\
		4 & 96.57 & 89.97 & 85.14 \\
		5 & 96.49 & 90.81 & 88.33 \\
		6 & 98.45 & 92.25 & 90.10 \\
		7 & 97.13 & 91.33 & 85.11 \\
		8 & 96.32 & 90.65 & 82.56 \\
		9 & 97.18 & 91.02 & 83.45 \\
		10 & 96.39 & 89.44 & 86.91 \\
		\hline
		mean$\pm$std
		& 96.88 & 90.78 & 87.24 \\
		\hline
	\end{tabular}
	\label{tab:ct}
\end{table}

\begin{figure}[!htbp]
	\centering
	\begin{tikzpicture}[
		circ/.style={draw, circle, thick, fill=gray!20, inner sep=2pt}, 
		box/.style={draw=red, rounded corners, align=center, minimum width=9.5cm, minimum height=3.5cm,line width=.03cm},
		box1/.style={draw=cyan, rounded corners, align=center, minimum width=9.5cm, minimum height=3.7cm,line width=.03cm},
		box2/.style={draw=green, rounded corners, align=center, minimum width=9.5cm, minimum height=3.7cm,line width=.03cm},
		dashedbox/.style={draw=blue, dashed, rounded corners, minimum width=12.9cm, minimum height=4.5cm, line width=.03cm},
		dashedbox1/.style={draw=orange, dashed, rounded corners, minimum width=12.9cm, minimum height=4.5cm, line width=.03cm},
		]
		\node[] (ct1){
			\includegraphics[width=.75\textwidth]{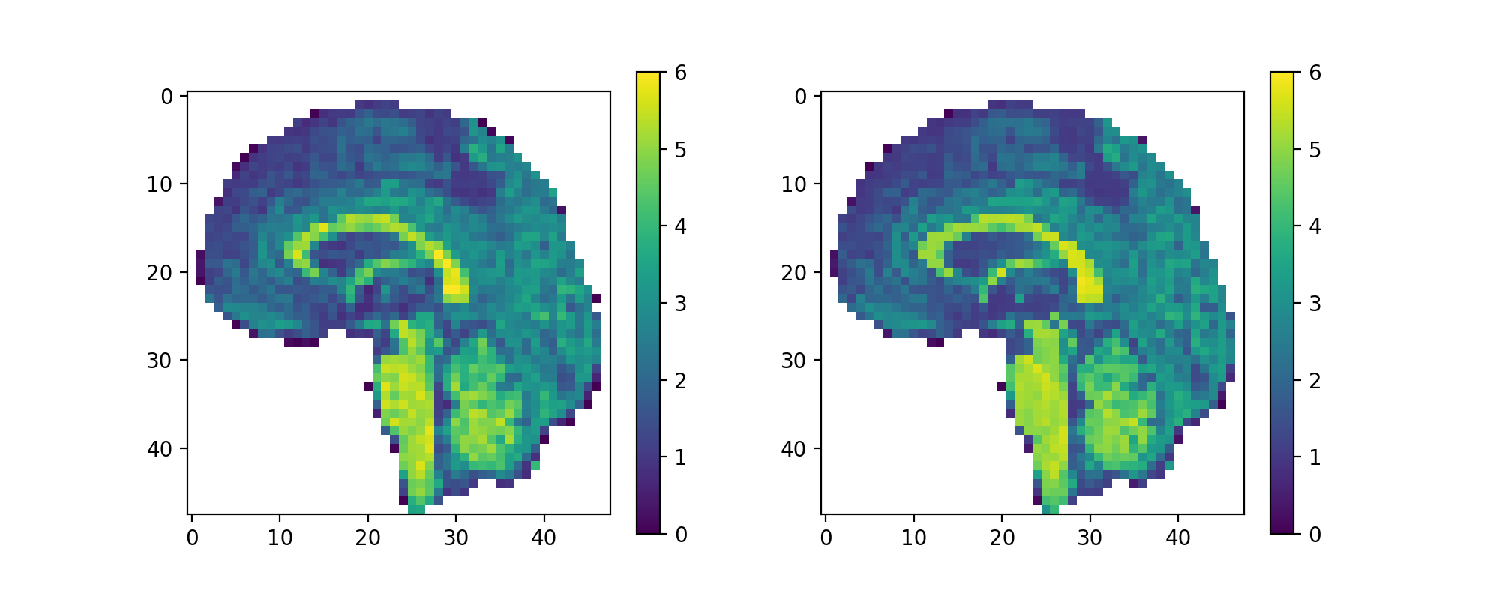}};
		\node[below=0cm of ct1] (ct2) {\includegraphics[width=.75\textwidth]{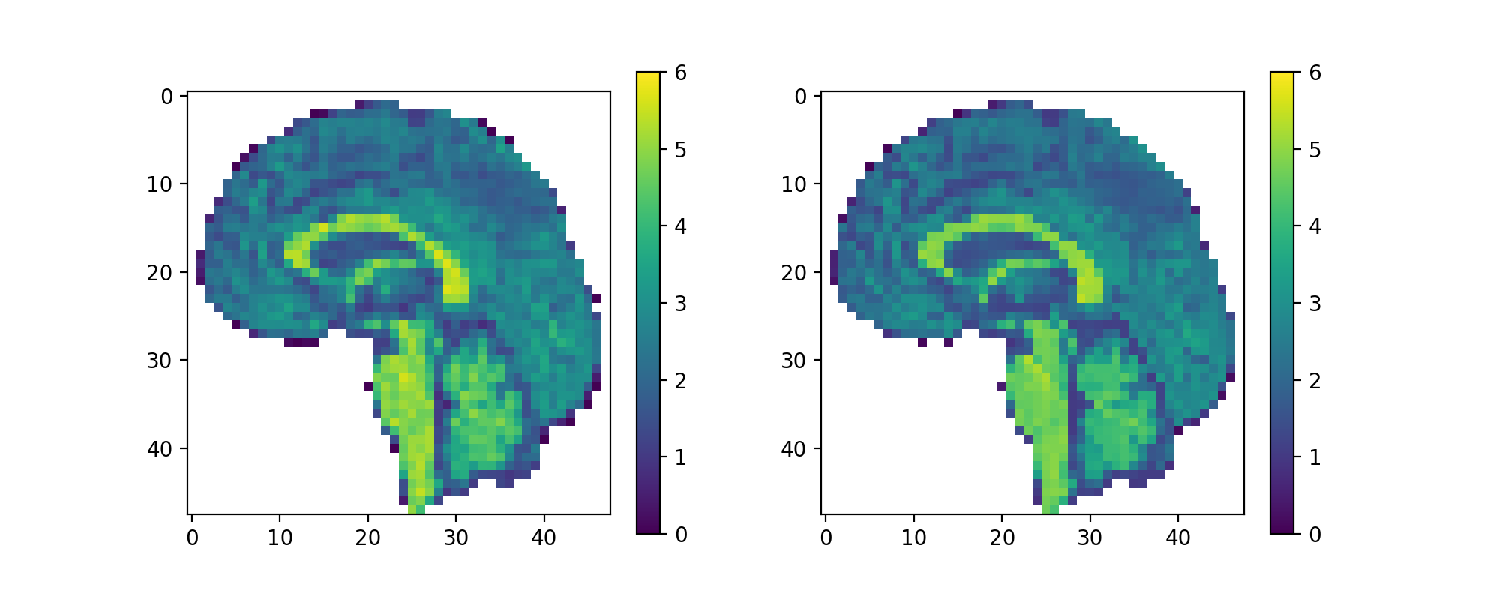}};
		\node[below=0.3cm of ct2] (ct3) {\includegraphics[width=.75\textwidth]{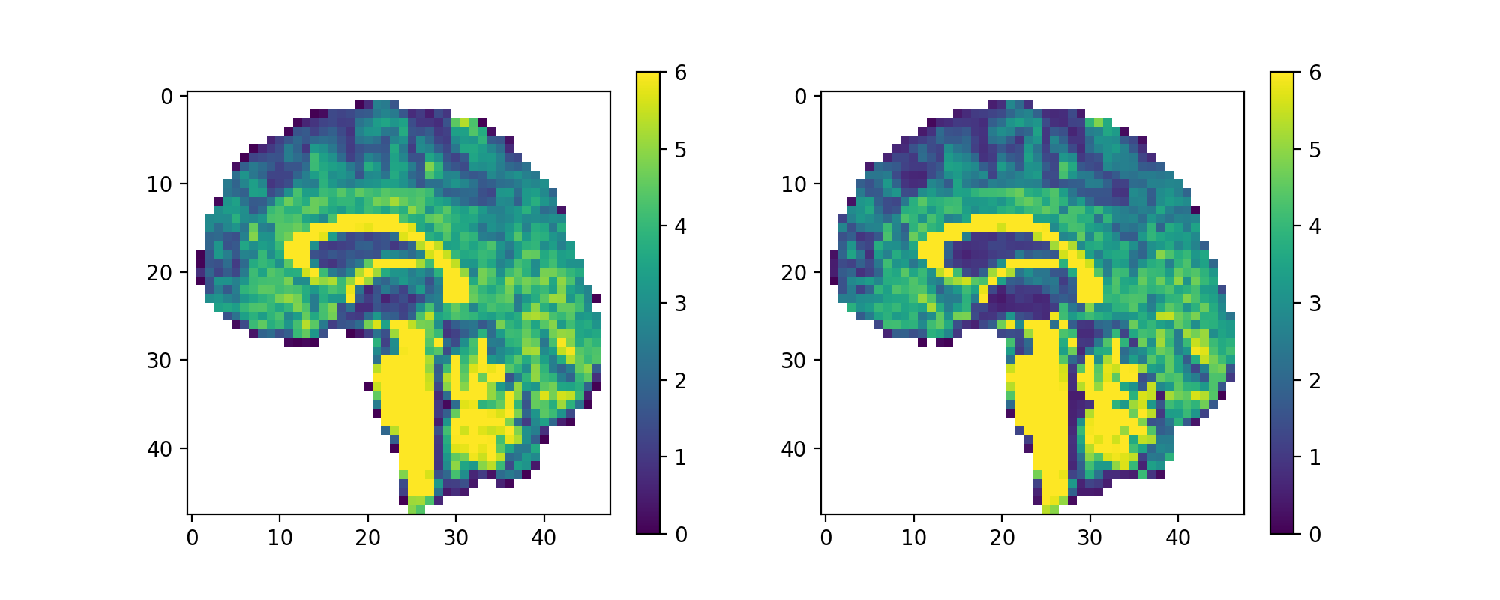}};
		\node[box, below=of ct1,xshift=-.1cm, yshift=5.6cm] (train) {{\color{red}Learning}\\\\\\\\\\\\\\\\};
		\node[box1, below=of ct2,xshift=-.1cm,  yshift=4.9cm] (transfer) {\\\\\\\\\\\\\\\\{\color{cyan}\ Transfer\quad Learning}};
		\node[box2, below=of ct3,xshift=-.1cm,  yshift=4.9cm] (predict) {\\\\\\\\\\\\\\\\{\color{green}Prediction Result}};
		\node[dashedbox, below=of ct1 ,xshift=-4.7cm,yshift=-1.45cm,rotate=90,align=left] (learn) {};
		\node[below=of ct1, xshift=-4.7cm,yshift=-1.3cm,rotate=90] (result) {\color{blue}LE-NO Results};
		\node[dashedbox1, below=of ct1 ,xshift=0cm,yshift=-1.45cm,rotate=90,align=left] (truth) {};
		\node[below=of ct1, xshift=4.5cm,yshift=-1.3cm,rotate=270] (result1) {\color{orange}Ground Truth};
	\end{tikzpicture}
	\caption{Training result (top) and transfer training comparison for another patient  (middle) and prediction result (bottom) : LE--NO results (left) and ground truth (right).}
	\label{fig:cttransfer}
\end{figure}

\section{Discussion and Conclusion}

We introduced the LE-NO, a principled framework for learning nonlinear reaction--diffusion dynamics that integrates spectral theory with neural operator learning. By leveraging Laplacian eigenfunctions on arbitrary geometries, LE-NO addresses key challenges in data efficiency, interpretability, and physical fidelity.

A core innovation of LE-NO is its use of problem-specific spectral bases formed by Laplacian eigenfunctions. This design offers three major advantages over existing neural operator architectures. First, the eigenfunctions diagonalize the diffusion operator, enabling stable and efficient time integration---particularly in regimes with small diffusion coefficients, where conventional models often become unstable. This property is evident in the Allen--Cahn example, where LE-NO maintains stability without artificial regularization. Second, unlike heuristic Fourier modes or learned data-driven bases, Laplacian eigenfunctions are inherently tied to the domain geometry and boundary conditions, reducing the need for additional architectural complexity. Third, the spectral basis naturally filters noise by truncating higher modes, preserving dominant dynamics while enhancing robustness---a feature especially valuable when working with sparse and noisy biomedical data, as seen in the Alzheimer’s application.

LE-NO’s loss function further combines data-driven terms with a physics-informed residual term that aligns training with the underlying PDE structure. Both theoretical and empirical results show that this formulation enhances convergence and accuracy, particularly in learning nonlinear interactions. The residual loss plays a critical role in ensuring physical consistency, as illustrated in benchmarks with derivative-driven dynamics that tend to destabilize conventional models.

Another important feature of LE-NO is its interpretability. In the Alzheimer’s application, the learned modes localized to anatomically distinct cortical regions, revealing spatial patterns consistent with known disease progression pathways. This arises naturally from the spectral representation, in which the neural network learns modal coefficients in the eigenfunction basis, effectively acting as a nonlinear filter on physiologically meaningful modes.

Such interpretability is particularly important in biomedical systems, where black-box models are often difficult to validate mechanistically. LE-NO addresses a key limitation of symbolic regression approaches such as SINDy, which struggle when the reaction terms lack closed-form expressions. By contrast, LE-NO learns the operator directly, capturing multiscale dynamics even when the governing physics are only partially known.

LE-NO generalizes effectively across a variety of settings:
\begin{itemize}
	\item \textbf{Boundary Conditions:} Handles both Dirichlet and Neumann conditions, including inhomogeneous cases such as in the KPP-Fisher example;
	\item \textbf{Irregular Domains:} Adapts to complex geometries, such as brain cortical surfaces, achieving over 90\% transfer accuracy across patient-specific meshes.
\end{itemize}

A notable computational benefit is the decoupling of eigenfunction computation (offline) from neural operator learning (online). While computing eigenfunctions involves a one-time cost, modern finite element solvers make this step efficient even in three dimensions. For example, in the Alzheimer’s case, brain-domain eigenfunctions were computed in minutes and subsequently enabled real-time inference across patient cohorts.

Despite its advantages, LE-NO has several limitations. Long-term prediction accuracy may degrade in systems with complex oscillatory or phase behavior (e.g., the Schrödinger equation), indicating the need for more advanced time integration or error correction strategies. Additionally, while LE-NO outperforms Fourier-based and kernel-based models across our benchmarks, its comparative performance against transformer-based architectures remains an open question. Moreover, extending the framework to convection-dominated and more general classes of PDEs will require adapting the spectral basis to reflect the governing dynamics more faithfully.

In conclusion, LE-NO bridges spectral theory with neural operator learning, yielding a framework that is both interpretable and adaptable to irregular domains. Its ability to embed PDE-based structure into learning makes it particularly well-suited for data-limited applications in biomedicine, where mechanistic understanding remains essential. By unifying physics-informed priors with modern machine learning, LE-NO paves the way for data-driven discovery of partial differential equations.

\section*{Data Availability}
Data for examples 1--4 are are generated using FeEniCS software (version 2019.1.0) and are available on GitHub at \href{https://github.com/Dong953/LENO}{https://github.com/Dong953/LENO}. Access to the ADNI dataset is publicly available via \href{http://adni.loni.usc.edu}{http://adni.loni.usc.edu} \cite{weiner2013alzheimer}.
\section*{Code Availability}
The source code for LE-NO is available on GitHub at \href{https://github.com/Dong953/LENO}{https://github.com/Dong953/LENO}.

\section*{Acknowledgements}
This work was supported by National Institute of General Medical Sciences through grant 1R35GM146894.

\appendix
\renewcommand{\thetheorem}{\Alph{section}.\arabic{theorem}}
\renewcommand{\thetheorem}{\Alph{section}.\arabic{theorem}}
\renewcommand{\theremark}{\Alph{section}.\arabic{remark}}
\renewcommand{\thelemma}{\Alph{section}.\arabic{lemma}}
\renewcommand{\theequation}{\Alph{section}.\arabic{equation}}
\section{Convergence Results}\label{appdendix:approx}
To analyze the approximation properties of the neural operator $\mathcal{N}$, we consider a sequence of basis functions  $\{\phi_i^{\mathcal{X}}\}$ of a general Hilbert space $\mathcal{X}$ satisfying the following approximation property:
\begin{equation}
	\lim_{n\rightarrow\infty} \inf_{v_n\in {\rm span}\{\phi_i^{\mathcal{X}}\}_{i=1}^n} \|v-v_n\|_{\mathcal{X}} = 0,
\end{equation}
for all $v\in {\mathcal{X}}$. We define $\mathcal{T}_n^\mathcal{X}$ as the projection operators on $\text{span}~\{\phi_i^\mathcal{X}\}_{i=1}^n$.

Then given Hilbert spaces $\mathcal{X},\mathcal{Y}$ and a nonlinear operator $\mathcal{F}:\mathcal{X}\rightarrow \mathcal{Y}$, if these spaces have projection operators $T^\mathcal{X}_n,T^\mathcal{Y}_n$ associated with basis $\phi_i^\mathcal{X}\in X,\phi_i^Y\in \mathcal{Y}$, respectively. For a fixed integer $n$, we define an approximation operator for $\mathcal{F}$ as
\begin{equation}
	\mathcal{F}_n (u):=	T^\mathcal{Y}_n \circ \mathcal{F} \circ T^\mathcal{X}_n (u)
\end{equation}

Rewriting $T_n^\mathcal{Y}u$ as a linear combination of the basis $\{\phi_i^\mathcal{Y}\}_{i=1}^n$, we have 
\begin{equation}
	T_n^\mathcal{Y}y=\sum_{i=1}^{n}\alpha_i^\mathcal{Y}(y)\phi^\mathcal{Y}_i,\quad y\in \mathcal{Y},
\end{equation}
where $\alpha_i^\mathcal{Y}(u)$ are linear continuous functionals. 
The approximation operator $\mathcal{F}_n$  can then be expressed as
\begin{equation}
	\begin{aligned}
		\mathcal{F}_n(u) =& \sum_{j=1}^{n} \alpha^\mathcal{Y}_j(\mathcal{F}\circ T^\mathcal{X}_nu)\phi_j^\mathcal{Y}. 
	\end{aligned}
\end{equation}
Further, we rewrite $T^\mathcal{X}_Pu$ as a linear combination of $\phi^\mathcal{X}_i$
\begin{equation}
	T^\mathcal{X}_nu = \sum_{i=1}^{n}\alpha_i^\mathcal{X}(u)\phi^\mathcal{X}_i,\quad u\in \mathcal{X}.
\end{equation}
The coefficient part is given by
\begin{equation}
	\alpha^\mathcal{Y}_j(\mathcal{F}\circ T^\mathcal{X}_nu) = \alpha^\mathcal{Y}_j(\mathcal{F}\circ \sum_{i=1}^n\alpha_i^\mathcal{X}(u)\phi^\mathcal{X}_i).
\end{equation}
Here, $\alpha^\mathcal{Y}_j(\mathcal{F}\circ \sum_{i=1}^n\alpha_i^\mathcal{X}(u)\phi^\mathcal{X}_i).$ can be viewed as a function of $\alpha^\mathcal{X}_i(u)$. We define
\begin{equation}\label{eq:g}
	\mathcal{G}_j(\alpha_1^\mathcal{X}(u),\cdots,\alpha_n^\mathcal{X}(u)):=	\alpha^\mathcal{Y}_j(\mathcal{F}\circ \sum_{i=1}^n\alpha_i^\mathcal{X}(u)\phi^\mathcal{X}_i),
\end{equation}
where $\mathcal{G}_j$ is clearly a continuous function.

If we use a neural network to approximate $\mathcal{G}_j$ such that $$
\mathcal{G}_j(\cdot)\sim \sum_{i=1}^m\mathcal{A}_j^i\sigma (\omega_j^i\cdot+b_j^i),
$$
for some $\omega_j\in\mathbb{R}^n$ and $b_j^i\in\mathbb{R}$, we obtain the following neural operator
\begin{equation}
	\mathcal{N}(u):= \sum_{i=1}^n \sum_{j=1}^m \left({A}_j^i\phi_i^\mathcal{Y}\right) \sigma(\omega_i^i\cdot\bm{\alpha}_n^\mathcal{X}(u)+b_j^i).
\end{equation}
where $\bm{\alpha}^\mathcal{X}_n(u) = (\alpha^\mathcal{X}_1(u),\cdots,\alpha^\mathcal{X}_n(u))^T$. 

Here 
\begin{itemize}
	\item $\mathcal{W}_j^iu:=\omega_j^i\cdot \bm{\alpha}^\mathcal{X}_n(u)$ is a linear continuous operator on $\mathcal{X}$.
	\item $\mathcal{B}_j^i:=b_j^i\in \mathcal{X}$ is a bias term.
	\item $\mathcal{A}_j^iu:={A}_j^i\phi_i^\mathcal{Y}u,~u\in\mathcal{X}$ is a continuous linear operator from $\mathcal{X}$ to $\mathcal{Y}$.
\end{itemize}
If we sum the two summations together with $P = nm$, the neural operator can be rewritten in the compact form shown in \eqref{eq:no1}
\begin{equation}
	\mathcal{N}(u)=\sum_{i=1}^P \mathcal{A}_i\sigma(\mathcal{W}_iu+\mathcal{B}_i)
\end{equation}
Specifically, when using a neural network as the basis $\phi_i^\mathcal{V}$, the operator $\mathcal{A}_i$ takes the form shown in \eqref{eq:no2}.

To provide an approximation result for $\mathcal{N}$, it suffices to analyze the approximation properties of $\mathcal{F}_n$. Since the difference between $\mathcal{N}$ and $\mathcal{F}_n$ can be made arbitrarily small by choosing a sufficiently wide neural network (large $m$), which is guaranteed by the universal approximation property of neural networks.

To begin with, we consider an extended set for the approximation result, defined as:
\begin{equation}\label{eq:extendK}
	\widehat{K} = K \cup \overline{\bigcup_{j=1}^{\infty}T^{\mathcal{X}}_{j}(K)}.
\end{equation}
We now establish the compactness of the extended set $\widehat{K}$.

\begin{lemma}[compactness of extended set]\label{lem:ext}
	The extended set $\widehat{K}$ given in (\ref{eq:extendK}) is compact.
\end{lemma}

\begin{proof}
	Since  $T^\mathcal{X}_n$ is linear continuous, it is uniformly continuous on the compact set $K$. For any $r>0$, there exists ${N}\in\mathbb{N}$ such that 
	$$
	\sup_{u\in \mathcal{X}}\|u-T^\mathcal{X}_ju\|_\mathcal{X}<r/2,
	$$
	for $j\ge {N}$. Now, cover $K$ by finite balls of radius $r/2$ and let $z_1,\cdots,z_l$ be the centers of these balls. Then balls of radius $r$ around $z_1,\cdots,z_l$ cover
	\begin{equation}
		K\cup \bigcup_{j={N}}^\infty {T}^{\mathcal{X}}_j(K),
	\end{equation}
	since for any $z\in K$, there exists $z_l$ such that $\|z-z_l\|_\mathcal{X}<r/2$ and 
	\begin{equation}
		\|T^\mathcal{X}_j(z)-z_l\|_\mathcal{X}\le \|T^\mathcal{X}_j(z)-z\|_\mathcal{X}+\|z-z_l\|_{\mathcal{X}}\le r/2+r/2<r.
	\end{equation}
	The remaining set
	\begin{equation}
		S = \bigcup_{j=1}^{{N}-1}{T}^\mathcal{X}_j(K),
	\end{equation}
	is a finite union of images of the compact set $K$ under the linear operators $T^\mathcal{X}_j$, and is therefore compact. Thus it can also be covered by finitely many balls of radius $r$. Since $r>0$ is arbitrary and $\widehat{K}$ is closed, the compactness follows.
\end{proof}

First, we consider using the projection operator $\Pi_P$ on $V_P$ to give the specific formulation $\mathcal{F}_P = \Pi_P \circ \mathcal{F} \circ \Pi_P$. The following theorem establishes the approximation property for $\mathcal{F}_P$.

\begin{lemma}\label{lem:approxop}
	Let $\mathcal{F}: \mathcal{X}(\Omega)\cap H^k(\Omega)\subset \mathcal{Y}(\Omega)$ be  a Lipschitz continuous operator with constant $L$,  $K\subset U$ be a compact set and $\mathcal{F}(K)\subset H^k(\Omega)$, it holds that
	\begin{equation}
		\sup_{u\in K} \|\mathcal{F}_P(u)-\mathcal{F}(u)\|_{L^2}\le C(L+1)P^{-{k\over d}}.
	\end{equation}
	given $P \in \mathbb{N}$.
\end{lemma}

\begin{proof}
	
	Substituting the formulation of $\mathcal{F}_P$, we have
	$$
	\|\mathcal{F}(u)-\Pi_P \circ \mathcal{F} \circ \Pi_P(u)\|_{L^2}\le \|\mathcal{F}(u)-\Pi_P\circ \mathcal{F}(u)\|_{L^2} + \|\Pi_P\circ( \mathcal{F}(u)-\mathcal{F}\circ \Pi_P (u))\|_{L^2},
	$$
	Since $\Pi_P$ is  linear continuous operator and $\mathcal{F}$ is a continuous operator, the set $\mathcal{F}(K)$  is also compact. By the approximation property of the basis $\phi_i$ , it holds that
	$$
	\|\Pi_P v-v\|_{L^2} \le CP^{-{k\over d}}, \quad \forall v\in \mathcal{F}(K).
	$$
	Refer to \cite[Theorem 4.10]{fabian2011banach}, the projection operators are uniform bounded, i.e., there exists $C$ independent of $P$ such that
	$$
	\|\Pi_P\|\le C.
	$$
	Thus, we have :
	$$
	\begin{aligned}
		\|\Pi_P\circ( \mathcal{F}(u)-\mathcal{F}\circ \Pi_P(u))\|_{L^2} &\le C\|\mathcal{F}(u)-\mathcal{F}\circ \Pi_P (u)\|_{L^2}\\
		&\le CL\|u-\Pi_P (u)\|_{L^2}\\
		&\le CLP^{-{k\over d}}.
	\end{aligned}
	$$
	Collecting the above results, we achieve the desired results.
\end{proof}

Combining the above results, we now provide the proof of Theorem \ref{thm:approxop}.
\begin{proof}[proof of Theorem \ref{thm:approxop}]
	Using the triangle inequality, we have
	$$
	\|\mathcal{N}(u)-\mathcal{F}(u)\|_{\mathcal{Y}}\le \|\mathcal{N}(u)-\mathcal{F}_P(u)\|_{L^2} + \|\mathcal{F}_P(u)-\mathcal{F}(u)\|_{L^2}
	$$
	As shown in Lemma \ref{lem:approxop}, we have
	$$
	\sup_{u\in K}\|\mathcal{F}_P(u)-\mathcal{F}(u)\|_{\mathcal{Y}}\le C(L+1)P^{-{k\over d}}.
	$$
	Meanwhile, recalling the definition of the continuous function $\mathcal{G}_j$ in \eqref{eq:g},  we can express
	$$
	\mathcal{F}_P(u)-\mathcal{N}(u)= \sum_{j=1}^P (\mathcal{G}_j(\bm{\beta}(u))-\mathcal{G}_j(\bm{\beta}(u);\theta))\phi_j,
	$$
	where $\mathcal{G}_j(\cdot;\theta)$ is defined as
	$$\mathcal{G}_j(\circ;\theta):=\sum_{i=1}^W A_j^i\sigma(\omega_j^i\cdot\circ+b_j^i)$$ 
	with $\omega_j^i\in\mathbb{R}^P,~A_j^i,b_j^i\in \mathbb{R}$ and $\bm{\beta}(u)=((\phi_1,u),\cdots,(\phi_P,u))^T$. 
	
	Since $\mathcal{F}'$ is bounded, it follows from the definition that $\mathcal{G}_j\in W^{1,\infty}$. Consider ReLU as activation function, using \cite[Theorem 1]{mao2023rates}, it holds that
	
	$$
	\|\mathcal{G}_j-\mathcal{G}_j^\theta\|_{L^\infty} \le C\|\mathcal{G}_j\|_{W^{1,\infty}}\log(m)^{{1\over 2}+P}m^{-{1\over P}{P+2\over P+4}}.
	$$
	This implies
	$$
	\sup_{u\in K}\|\mathcal{F}_P(u)-\mathcal{N}(u)\|_{L^2} \le C(\sum_{j=1}^P\|\mathcal{G}_j\|_{W^{1,\infty}}^2)^{1/2}\log(m)^{{1\over 2}+P}m^{-{1\over P}{P+2\over P+4}}.
	$$
	Note that 
	$$
	\begin{aligned}
		\sum_{j=1}^P\|\mathcal{G}_j(\bm{\beta}(u)\|_{W^{1,\infty}}^2
		&\le \sum_{j=1}^\infty\|\mathcal{G}_j(\bm{\beta}(u)\|_{W^{1,\infty}}^2 \\
		& =\sup_{u\in K} \sum_{j=1}^\infty|\mathcal{G}_j'(\bm{\beta}(u)|^2+ |\mathcal{G}_j(\bm{\beta}(u)|^2\\
		& =\sup_{u\in K} \|\sum_{j=1}^\infty\mathcal{G}_j'(\bm{\beta}(u))\phi_j\|_{L^2}^2+ \|\sum_{j=1}^\infty\mathcal{G}_j(\bm{\beta}(u))\phi_j\|_{L^2}^2\\
		&=\sup_{u\in K} \|\mathcal{F}'(\Pi_Pu)\|_{L^2}^2 + \|\mathcal{F}(\Pi_Pu)\|_{L^2}^2 \le C(K,L),
	\end{aligned}
	$$
	for some constant $C$ depends on compact set $K$ since $\mathcal{F}$ 
	is uniform continuous and $\|\mathcal{F}'\|\le L$. In this way, we have
	$$
	\sup_{u\in K}\|\mathcal{F}_P(u)-\mathcal{N}(u)\|_{L^2} \le C\log(m)^{{1\over 2}+P}m^{-{1\over P}{P+2\over P+4}}.
	$$
	Combining the above results, we obtain
	$$
	\sup_{u \in K} \|\mathcal{N}(u) - \mathcal{F}(u)\|_{L^2} \leq C(L+1)P^{-{k\over d}}+C\log(m)^{{1\over 2}+P}m^{-{1\over P}{P+2\over P+4}}.
	$$
\end{proof}
Now, to analyze the error of the approximate solution $\widetilde{u}_P^n$, we consider the continuous variation problem of using  neural operator $\mathcal{N}$ to approximate $\mathcal{F}$, which gives the following problem: Find $\widetilde{u}\in L^2(0,T;H_0^1(\Omega))$ with $\widetilde{u}_t\in L^2(0,T;H^{-1}(\Omega))$ such that
\begin{equation}
	\label{eq:neuralreplace}
	(\widetilde{u}_t,v) + (D\nabla \widetilde{u},\nabla v) = (\mathcal{N}(\widetilde{u}),v)\quad t\in (0,T), ~ \forall v \in H_0^1(\Omega).
\end{equation}
for $ v \in H_0^1(\Omega)$ with $\widetilde{u}(x,0) = u_0(x)$.

\begin{lemma}\label{lem:err-sol}
	Let $u$ be the solution of \eqref{eq:baseeq} and $\widetilde{u}$ be the solution of \eqref{eq:neuralreplace}. Under the assumption of Theorem \ref{thm:approxop}, we have the following estimate
	\begin{equation}
		\underset{0 \leq t \leq T}{ \sup }	\|u(t)-\widetilde{u}(t)\|_{H_0^1(\Omega)} \le C\sup_{u}\|\mathcal{F}(u)-\mathcal{N}(u)\|_{L^2} T^{1/2}\exp \big(CLT\big).
	\end{equation}
\end{lemma}

\begin{proof}
	Following the regularity results of elliptic equations \cite{evans2022partial}, the following inequality holds:
	\begin{equation}
		\begin{aligned}
			\underset{0 \leq t \leq T}{ \sup }\|u(t)-\widetilde{u}(t)\|_{H^1(\Omega)}
			& \leq C\|\mathcal{F}(u)-\mathcal{N}(\widetilde{u})\|_{L^2\left(0, T ; L^2(\Omega)\right)}\\
			& \leq C\|\mathcal{F}(\widetilde{u})-\mathcal{N}(\widetilde{u})\|_{L^2\left(0, T ; L^2(\Omega)\right)}+C\|\mathcal{F}(u)-\mathcal{F}(\widetilde{u})\|_{L^2\left(0, T ; L^2(\Omega)\right)}\\
			&\leq C T^{1/2}\sup_{u}\|\mathcal{F}(u)-\mathcal{N}(u)\|_{L^2}+CL\|u-\widetilde{u}\|_{L^2(0, T ; L^2(\Omega))}\\
			&\leq C T^{1/2}\sup_{u}\|\mathcal{F}(u)-\mathcal{N}(u)\|_{L^2}+CL\|u-\widetilde{u}\|_{L^2(0, T ; H^1_0(\Omega))},
		\end{aligned}
	\end{equation}
	Utilizing Gronwall’s inequality on the above result, we obtain
	\begin{equation}
		\underset{0 \leq t \leq T}{ \sup }	\|u(t)-\widetilde{u}(t)\|_{H^1(\Omega)} \le C\sup_{u}\|\mathcal{F}(u)-\mathcal{N}(u)\|_{L^2} T^{1/2}\exp \big(CLT\big).
	\end{equation}
\end{proof}

\begin{remark}
	This theorem highlights that the error may grow as time progresses. To achieve a more accurate solution using the neural operator, it is essential for the neural operator to provide a more precise approximation of the nonlinear term.
\end{remark}

Consider time step $t_n=n\tau$ ($n = 0, 1,\cdots N$) be a uniform partition of $[0, T]$ with $\tau = T/N$, we then use semi-implicit Euler scheme for time discretization of \eqref{eq:neuralreplace}: find $\widetilde{w}^n(x)\in H^1_0(\Omega)$ such that
\begin{equation}\label{eq:semidiscret}
	(\frac{\widetilde{w}^{n}-\widetilde{w}^{n-1}}{\tau},v) +(D\nabla \widetilde{w}^{n},\nabla v) = (\mathcal{N}(\widetilde{w}^{n-1}),v),
\end{equation}
for any $v\in H_0^1(\Omega)$ and $n=1,2,\cdots,N$ with $\widetilde{w}^0(x) = u^0(x)$.

For the discrete variational problem, the approximation property of the space $V_P$ is: Let $u\in H^k(\Omega)\cap H_0^1(\Omega)$, it holds that
\begin{equation}
	\min_{u_P\in V_P}	\| u -u_P\|_{L^2} = \mathcal{O}(P^{-{k\over d}}).
\end{equation}

Through careful analysis under the regularity requirements for $\mathcal{N}$, we can derive the following  estimate.

\begin{lemma}\label{lem:err-dis}
	Let  $\widetilde{u}$ be solution of \eqref{eq:neuralreplace} satisfying $\widetilde{u}\in L^2(0,T;H_0^1(\Omega))\cap L^2(0,T;H^k(\Omega))$ and $\widetilde{u}^n$ be solution of \eqref{eq:neuralreplace} restricted on $V_P$, it holds that
	\begin{equation}\label{eq:errorestimate}
		\sup_{0\le n\le N}	\|\widetilde{u}(t_n)-\widetilde{u}^n\|_{L^2(\Omega)} \le C(\tau+P^{-{k\over d}}).
	\end{equation}
	where $C$ is is independent of $\tau $ and $P$. 
\end{lemma}

The estimate is fairly standard, relying on the analysis of the explicit Euler time-stepping method and the approximation properties of the finite-dimensional eigenfunction space $V_P$, combined with the triangle inequality applied to $\widetilde{u}(t_n)-\widetilde{w}^n$ and $\widetilde{w}^n-\widetilde{u}^n$. For brevity, we omit the proof and refer to \cite{thomee2007galerkin} for a similar argument.

Now we can give the proof for Theorem \ref{thm:err-sol}.
\begin{proof}[proof of Theorem \ref{thm:err-sol}]
	Using the triangle inequality for Lemma \ref{lem:err-sol} and \ref{lem:err-dis}, there exists a neural operator $\mathcal{N}$ such that
	$$
	\begin{aligned}
		\|\widetilde{u}^n-u(t_n)\|_{L^2} &\le \|\widetilde{u}^n-\widetilde{u}(t_n)\|_{L^2} + \|\widetilde{u}(t_n)-u(t_n)\|_{L^2}\\
		&\le C(\tau+P^{-{k\over d}}) + C\sup_{u}\|\mathcal{F}(u)-\mathcal{N}(u)\|_{L^2} T^{1/2}\exp(CLT).
	\end{aligned}
	$$
	Here we achieve the desired result.
\end{proof}

We also have the model approximation results, namely, the error between the PDE residual $\mathcal{R}_P^n(u)$ and the nonlinear term $\mathcal{F}(u(t_n))$.
\begin{lemma}\label{thm:residual}
	Assume the solution of problem \eqref{eq:baseeq} satisfy $u_{tt}\in H^k(\Omega)$ and $\Delta u \in H^k(\Omega)$. The PDE residual $\mathcal{R}_P^n(u)$ satisfies
	\begin{equation}
		\|\mathcal{R}^n(u)\cdot \Phi_P-\mathcal{F}(u(t_n))\|_{L^2}\le C(\tau + P^{-{k\over d}}),
	\end{equation}
	where $\Phi_P=(\phi_1,\cdots,\phi_P)^T$ and $\tau=\max_n(t_n-t_{n-1})$.
\end{lemma}

\begin{proof}
	We start with the governing equation:
	$$
	u_t(t_n)-D\Delta u(t_n) = \mathcal{F}(u(t_n)).
	$$
	Let  $\Pi_P$ denote the projection onto $V_P$ , the approximation property ensures that
	$$
	\|(I-\Pi_P)\left(u_t(t_n) -D\Delta u(t_n) \right)\|_{L^2} \le CP^{-{k\over d}},
	$$ 
	since $u_t,\Delta u\in H^k(\Omega$). Using the forward Euler scheme for the time derivative, we have the estimate
	$$
	|({u(t_n)-u(t_{n-1})\over \tau}-u_t(t_n),\phi_i)| \le C\tau |(u_{tt},\phi_i)|.
	$$
	This implies 
	$$
	\|\Pi_P (u_t(t_n) - {u(t_n)-u(t_{n-1})\over \tau})\|^2 \le C\tau\sum_{i=1}^P |(u_{tt},\phi_i)|^2 \le C\tau^2 \|u_{tt}(t_n)\|_{L^2}^2.
	$$
	From the definition of $\mathcal{R}_P^n$, it follows that
	$$
	\mathcal{R}_P^n\cdot\Phi_P=\Pi_P \left({u(t_n)-u(t_{n-1})\over \tau}-D\Delta u(t_n)\right).
	$$
	Finally, applying the triangle inequality and combining the results above, we obtain
	$$
	\begin{aligned}
		\|\mathcal{R}_P^n\cdot\Phi_P-\mathcal{F}(u(t_n))\|_{L^2} &=\|\mathcal{R}_P^n\cdot\Phi_P-(u_t(t_n)-D\Delta u(t_n))\|_{L^2}\\
		&\le \|\mathcal{R}_P^n\cdot\Phi_P - \Pi_P (u_t(t_n)-D\Delta u(t_n))\|_{L^2}\\
        &\quad+ \|(I-\Pi_P )(u_t(t_n)-D\Delta u(t_n))\|_{L^2}\\
		&\le C(\tau +P^{-{k\over d}}).
	\end{aligned}
	$$
\end{proof}
This result shows that our residual term $\mathcal{R}$ provides a good approximation of the nonlinear term $\mathcal{F}$. Consequently, when considering the residual loss $L^R$ in the optimization process, the trained neural operator $\mathcal{N}$ can efficiently approximate $\mathcal{F}$.

\begin{remark}
	In practice, the available data may have limited regularity, such as belonging to $L^\infty(\Omega)$. By the Sobolev embedding theorem  \cite{adams2003sobolev}, it holds that $L^\infty(\Omega) \hookrightarrow H^k(\Omega)$ for $k<{d\over 2}$ for $k<{d\over 2}$, thus we have 
	\begin{equation}
		\|\mathcal{R}^n(u)\cdot \Phi_P-\mathcal{F}(u(t_n))\|_{L^2}\le C(\tau + P^{-{1\over 2}+\xi}),
	\end{equation}
	for arbitrary $\xi >0$. For a uniform time step $\tau = \frac{T}{N}$, with sufficiently large $N$ and $P$, the error can be made arbitrarily small.
\end{remark}

In this way, given data $\mathcal{R}^n(u)$, there exists a neural operator $\mathcal{N}$ such that
\begin{equation}
	\begin{aligned}
		\|\mathcal{R}^n(u)\cdot\Phi_P-\mathcal{N}(u)\|_{L^2} 
		&\le C(\tau+P^{-{k\over d}}) +C\log(m)^{{1\over 2}+P}m^{-{1\over P}{P+2\over P+4}}.
	\end{aligned}
\end{equation}

\section{Optimization Analysis}\label{appdendix:opt}
The proof of Theorem~\ref{thm:opt} is divided into two parts, corresponding to the two cases presented in the theorem.
\begin{proof}[proof of Theorem \ref{thm:opt}]

Consider only the residual loss $L^R$, we consider the following estimate using the absolute squared loss for $M$ samples with $N$ time steps
\begin{equation}
    L^R(\theta):=  {1\over NM}\sum_{n=1}^N\sum_{i=1}^M \|\mathcal{R}^n(u_i)- \mathcal{G}({\bm{\beta}}^{n-1}({u}_i;\theta))\|_{l^2}^2.
\end{equation}

    Without loss of generality, we assume $D=1$. We define the function
    \begin{equation}
        L_i^n(\theta) = \|\mathcal{R}^n(u_i)- \mathcal{G}({\bm{\beta}}^{n-1}({u}_i;\theta))\|_{l^2}^2.
    \end{equation}
    The loss function can be rewritten as
	\begin{equation}
		L(\theta) = {1\over MN}\sum_{n=1}^N\sum_{i=1}^M L_i^n(\theta).
	\end{equation}
	Consider $\mathcal{G} : \mathbb{R}^P \to \mathbb{R}^P$ as a shallow neural network with width $m$ and examine its behavior under gradient flow with respect to its parameters $\theta$, using the loss function $L(\theta)$:
	\begin{equation}
		\frac{d}{ds}\theta(s)=-\nabla_\theta L(\theta(s)).
	\end{equation}
	The neural Tangent Kernel (NTK) \cite{jacot2018neural} describes the dynamics of this process in the infinite-width limit $m\rightarrow\infty$. First, we have
	\begin{equation}
		\frac{d}{ds}\mathcal{G}(\bm{\beta}_i^q;\theta(s)) = -{1\over NM}\nabla_\theta \mathcal{G}(\bm{\beta}_i^q;\theta(s))\sum_{n=1}^N\sum_{j=1}^M  \nabla_\theta \mathcal{G}(\bm{\beta}_j^n;\theta(s))^T\partial_{\mathcal{G}}{L}_j^n.
	\end{equation}
	The standard NTK of $\mathcal{G}$ for $\bm{x},\bm{y}\in\mathbb{R}^P$ is defined as
	\begin{equation}
		\nabla_\theta \mathcal{G}(\bm{y};\theta)  \nabla_\theta \mathcal{G}(\bm{x};\theta)^T \rightarrow \mathbb{E}_{\theta\sim \mu}(\nabla_\theta \mathcal{G}(\bm{y};\theta)  \nabla_\theta \mathcal{G}(\bm{x};\theta)^T )\in\mathbb{R}^{P\times P} .
	\end{equation}
	where $\mu$ denotes the initialization distribution of $\theta$.
	
	Note that for a stacked neural network , i.e., $\mathcal{G}_i(\bm{x})=\sum_{i=1}^m a_i^j\sigma(\omega_i^j\cdot x+b_i^j),~i=1,\cdots,P$, we have
	$$\mathbb{E}(\nabla_\theta\mathcal{G}_i(\bm{y})\nabla_\theta\mathcal{G}_j(\bm{x})^T)=0, \quad i\neq j. $$
	This implies the existence of a function $\gamma(\bm{y},\bm{x})=\mathbb{E}(\nabla_\theta\mathcal{G}_i(\bm{y})\nabla_\theta\mathcal{G}_i(\bm{x})^T)$ since $\mathcal{G}_i$ are i.i.d. and $s^*>0$ such that
	\begin{equation}
		\frac{d}{ds}\mathcal{G}(\bm{\beta}_i^q;\theta(s)) = -{1\over NM}\sum_{n=1}^N\sum_{j=1}^M  \gamma(\bm{\beta}_i^q,\bm{\beta}_j^n)\partial_{\mathcal{G}}{L}_j^n.
	\end{equation}
	for some $s\in [0,s^*]$ with high probability \cite{yang2025homotopy} as $m\rightarrow\infty$. For the iteration of the loss function, we have
	\begin{equation}
		\begin{aligned}
			{d\over ds}L(\theta(s)) &= {1\over NM}\sum_{n=1}^N\sum_{i=1}^M\partial_{\mathcal{G}}L_i^n\cdot \frac{d}{ds}\mathcal{G}=-{1\over N^2M^2}\sum_{q,n=1}^N\sum_{i,j=1}^M \partial_{\mathcal{G}}L_i^q\gamma(\bm{\beta}_i^q,\bm{\beta}_j^n)\partial_{\mathcal{G}}{L}_j^n.
		\end{aligned}
	\end{equation}
	
	Rewriting the right-hand side into matrix form
	\begin{equation}
		{d\over ds}L(\theta(s)) = -{1\over N^2M^2}e^TKe,
	\end{equation}
	where $K\in\mathbb{R}^{NMP\times NMP}$ and $e\in\mathbb{R}^{NMP}$ are defined as
	\begin{equation}
		K = I_P\otimes\Gamma , \quad \Gamma=(\gamma_{rs})\in\mathbb{R}^{NM\times NM}, \quad e = (e_1,e_2,\cdots,e_P)^T\in\mathbb{R}^{NMP}.
	\end{equation}
        with $\gamma_{rs}=\gamma(\bm{\beta}_i^q,\bm{\beta}_j^n)$, $ e_p = ((e_p)_s)\in\mathbb{R}^{NM}$ and $(e_p)_s=\partial_{\mathcal{G}}L_j^n$ for $r=(i-1)N+q,\ s=(j-1)N+n$, and $p=1,\cdots,P$.
	The eigenvalues of $K$ are determined by those of $\Gamma$, which is the NTK of a shallow neural network under an $L^2$ fitting setup. Its positive definiteness ensures that the minimum eigenvalue, $\gamma > 0$, exists \cite{jacot2018neural}. Thus we have
	\begin{equation}
		{d\over ds}L(\theta(s)) \le  -{\gamma\over M^2} e^Te=-{\gamma\over N^2M^2}\sum_{n=1}^N\sum_{j=1}^M |\partial_{\mathcal{G}}{L}_j^n|^2
	\end{equation}
	Expanding $\partial_{\mathcal{G}} L_j^n$ and using the definition of $\mathcal{R}_j^n$:
	\begin{equation}\label{eq:Ljn}
		\begin{aligned}
			\partial_{\mathcal{G}}L_j 
			& = -2(\mathcal{R}^n_j-\mathcal{G}(\bm{\beta}_j^{n-1};\theta))
		\end{aligned}
	\end{equation}
	This gives
	\begin{equation}
		|\partial_{\mathcal{G}}{L}_j^n|^2 = 4(\mathcal{R}^n_j-\mathcal{G}(\bm{\beta}_j^{n-1};\theta))^T(\mathcal{R}^n_j-\mathcal{G}(\bm{\beta}_j^{n-1};\theta))=4L_j^n.
	\end{equation}
	Noting $L={1\over NM}\sum_{n=1}^N\sum_{j=1}^M L_j^n$, we obtain
	\begin{equation}
		{d\over ds}L(\theta(s)) \le -{4\gamma\over NM} L.
	\end{equation}
	This implies 
	\begin{equation}
		L(\theta(s)) \le L(\theta(0))\exp(-{4\gamma\over NM}s).
	\end{equation}

When considering both residual loss and data $L^2$ loss for one-step, we consider the following estimates using the absolute squared loss for $M$ samples
\begin{equation}
	L^D(\theta):={1\over M}\sum_{i=1}^M \| \bm{\beta}^1(u^1_i)- \widetilde{\bm{\beta}}^1((\widetilde{u}^1_i);\theta)\|_{l^2}^2,
\end{equation}
\begin{equation}
	L^R(\theta):=  {1\over M}\sum_{i=1}^M \|\mathcal{R}^1(u_i)- \mathcal{G}({\bm{\beta}}^0({u}_i;\theta))\|_{l^2}^2.
\end{equation}
We have
	\begin{equation}
		\widetilde{\bm{\beta}}^1 = (1+\tau\Lambda_P)^{-1}(\widetilde{\bm{\beta}}^0+\tau\mathcal{G}(\widetilde{\bm{\beta}}^0)).
	\end{equation}
	Note that the initial condition gives $\widetilde{\bm{\beta}}^0=\bm{\beta}^0$, 
	we define the function 
	\begin{equation}
		L_i(\theta)=L_i^R(\theta)+L_i^D(\theta):= \|\mathcal{R}^1_i - \mathcal{G}(\bm{\beta}_i^0;\theta)\|_{l^2}^2 + \|\bm{\beta}_i^1-(1+\tau\Lambda_P)^{-1}(\bm{\beta}_i^0+\tau\mathcal{G}(\bm{\beta}_i^0;\theta)\|_{l^2}^2
	\end{equation}
	where $\bm{\beta}_i^n= \bm{\beta}^n(u_i^n)$ and $\mathcal{R}_i^1= \mathcal{R}^1(u_i^1)$ for $i=1,2,\cdots,M$ and $n=0,1$.
	The loss function can be rewritten as
	\begin{equation}
		L(\theta) = {1\over M}\sum_{i=1}^M L_i(\theta).
	\end{equation}
	Consider $\mathcal{G} : \mathbb{R}^P \to \mathbb{R}^P$ as a shallow neural network with width $m$ and examine its behavior under gradient flow with respect to its parameters $\theta$, using the loss function $L(\theta)$:
	\begin{equation}
		\frac{d}{ds}\theta(s)=-\nabla_\theta L(\theta(s)).
	\end{equation}
	The neural Tangent Kernel (NTK) \cite{jacot2018neural} describes the dynamics of this process in the infinite-width limit $m\rightarrow\infty$. First, we have
	\begin{equation}
		\frac{d}{ds}\mathcal{G}(\bm{\beta}_i^0;\theta(s)) = -{1\over M}\nabla_\theta \mathcal{G}(\bm{\beta}_i^0;\theta(s))\sum_{j=1}^M  \nabla_\theta \mathcal{G}(\bm{\beta}_j^0;\theta(s))^T\partial_{\mathcal{G}}{L}_j.
	\end{equation}
	The standard NTK of $\mathcal{G}$ for $\bm{x},\bm{y}\in\mathbb{R}^P$ is defined as
	\begin{equation}
		\nabla_\theta \mathcal{G}(\bm{y};\theta)  \nabla_\theta \mathcal{G}(\bm{x};\theta)^T \rightarrow \mathbb{E}_{\theta\sim \mu}(\nabla_\theta \mathcal{G}(\bm{y};\theta)  \nabla_\theta \mathcal{G}(\bm{x};\theta)^T )\in\mathbb{R}^{P\times P} .
	\end{equation}
	where $\mu$ denotes the initialization distribution of $\theta$.
	
	Note that for a stacked neural network , i.e., $\mathcal{G}_i(\bm{x})=\sum_{i=1}^m a_i^j\sigma(\omega_i^j\cdot x+b_i^j),~i=1,\cdots,P$, we have
	$$\mathbb{E}(\nabla_\theta\mathcal{G}_i(\bm{y})\nabla_\theta\mathcal{G}_j(\bm{x})^T)=0, \quad i\neq j. $$
	This implies the existence of a function $\gamma(\bm{y},\bm{x})=\mathbb{E}(\nabla_\theta\mathcal{G}_i(\bm{y})\nabla_\theta\mathcal{G}_i(\bm{x})^T)$ since $\mathcal{G}_i$ are i.i.d. and $s^*>0$ such that
	\begin{equation}
		\frac{d}{ds}\mathcal{G}(\bm{\beta}_i^0;\theta(s)) = -{1\over M}\sum_{j=1}^M  \gamma(\bm{\beta}_i^0,\bm{\beta}_j^0)\partial_{\mathcal{G}}{L}_j.
	\end{equation}
	for some $s\in [0,s^*]$ with high probability \cite{yang2025homotopy} as $m\rightarrow\infty$. For the iteration of the loss function, we have
	\begin{equation}
		\begin{aligned}
			{d\over ds}L(\theta(s)) &= {1\over M}\sum_{i=1}^M\partial_{\mathcal{G}}L_i\cdot \frac{d}{ds}\mathcal{G}=-{1\over M^2}\sum_{i,j=1}^M \partial_{\mathcal{G}}L_i\gamma(\bm{\beta}_i^0,\bm{\beta}_j^0)\partial_{\mathcal{G}}{L}_j.
		\end{aligned}
	\end{equation}
	
	Rewriting the right-hand side into matrix form
	\begin{equation}
		{d\over ds}L(\theta(s)) = -{1\over M^2}e^TKe
	\end{equation}
	where $K\in\mathbb{R}^{MP\times MP}$ and $e\in\mathbb{R}^{MP}$ are defined as
	\begin{equation}
		K = I_P\otimes\Gamma ,\quad \Gamma=(\gamma(\bm{\beta}_i^0,\bm{\beta}_j^0))\in\mathbb{R}^{M\times M}
	\end{equation}
	\begin{equation}
		e = (e_1,e_2,\cdots,e_P)^T\in\mathbb{R}^{MP} \quad \text{ with }\quad  e_i = (\partial_{\mathcal{G}_i}{L}_1,\cdots,\partial_{\mathcal{G}_i}{L}_M)^T\in\mathbb{R}^M.
	\end{equation}
	The eigenvalues of $K$ are determined by those of $\Gamma$, which is the NTK of a shallow neural network under an $L^2$ fitting setup. Its positive definiteness ensures that the minimum eigenvalue, $\gamma > 0$, exists \cite{jacot2018neural}. Thus we have
	\begin{equation}
		{d\over ds}L(\theta(s)) \le  -{\gamma\over M^2} e^Te=-{\gamma\over M^2}\sum_{j=1}^M |\partial_{\mathcal{G}}{L}_j|^2
	\end{equation}
	Expanding $\partial_{\mathcal{G}} L_j$ and using the definition of $\mathcal{R}_j^1$:
	\begin{equation}\label{eq:Lj}
		\begin{aligned}
			\partial_{\mathcal{G}}L_j &= 
			-2(\tau(1+\tau\Lambda_P)^{-1}(\bm{\beta}_j^1-(1+\tau\Lambda_P)^{-1}(\bm{\beta}_j^0+\tau\mathcal{G}(\bm{\beta}_j^0;\theta)))+ \mathcal{R}_j^1-\mathcal{G}(\bm{\beta}_j^0;\theta))\\
			& = -2(I+\tau(1+\tau\Lambda_P)^{-1})(\mathcal{R}^1_j-\mathcal{G}(\bm{\beta}_j^0;\theta))
		\end{aligned}
	\end{equation}
	This gives
	\begin{equation}
		|\partial_{\mathcal{G}}{L}_j|^2 = 4(\mathcal{R}^1_j-\mathcal{G}(\bm{\beta}_j^0;\theta))^T(I+\tau(1+\tau\Lambda_P)^{-1})^2(\mathcal{R}^1_j-\mathcal{G}(\bm{\beta}_j^0;\theta)).
	\end{equation}
	Given that $L_j = (\mathcal{R}^1_j-\mathcal{G}(\bm{\beta}_j^0;\theta))^T(\tau(1+\tau\Lambda_P)^{-1})^2(\mathcal{R}^1_j-\mathcal{G}(\bm{\beta}_j^0;\theta))+(\mathcal{R}^1_j-\mathcal{G}(\bm{\beta}_j^0;\theta))^T(\mathcal{R}^1_j-\mathcal{G}(\bm{\beta}_j^0;\theta))$, we have
	\begin{equation}
		\begin{aligned}
			|\partial_{\mathcal{G}}{L}_j|^2 =& 4L_j + 8 (\mathcal{R}^1_j-\mathcal{G}(\bm{\beta}_j^0;\theta))^T(\tau(1+\tau\Lambda_P)^{-1})(\mathcal{R}^1_j-\mathcal{G}(\bm{\beta}_j^0;\theta))\\
			\ge & 4L_j + 8 \frac{\tau(1+\tau\lambda_P)^{-1}}{(1+\tau^2(1+\tau\lambda_1)^{-2})} L_j.
		\end{aligned}
	\end{equation}
	Noting $L={1\over M}\sum_{j=1}^M L_j$, we obtain
	\begin{equation}
		{d\over ds}L(\theta(s)) \le -{\gamma\over M} (4+8 \frac{\tau(1+\tau\lambda_P)^{-1}}{(1+\tau^2(1+\tau\lambda_1)^{-2})}) L.
	\end{equation}
	This implies 
	\begin{equation}
		L(\theta(s)) \le L(\theta(0))\exp(-{\gamma\over M}(4+8 \frac{\tau(1+\tau\lambda_P)^{-1}}{(1+\tau^2(1+\tau\lambda_1)^{-2})})s).
	\end{equation}
\end{proof}

\section{Application to Various Boundary Conditions}\label{appendix:boundary}
\subsection{Inhomogeneous Dirichlet Boundary Condition}
For the inhomogeneous Dirichlet boundary condition 
\begin{equation}
	u=g \quad \text{on}~\partial\Omega\times(0,T).
\end{equation}
we consider harmonic extension of the boundary to obtain $u_g$,  which satisfies:
\begin{equation}
	\left\{
	\begin{aligned}
		-\Delta u_g &= 0, &\text{in}~\Omega \times (0,T)\\
		u_g &=g,&\text{on}~\partial\Omega\times(0,T)
	\end{aligned}
	\right.
\end{equation}
This harmonic extension can be derived using the finite element method. Using $u_g$, we reformulate the original problem in terms of $w:=u-u_g$ which satisfies:
\begin{equation}
	\left\{
	\begin{aligned}
		w_t-D\Delta w &= \mathcal{F}(w+u_g)-(u_g)_t, & \text{ in }\Omega\times [0,T]\\
		w(x,0) &= w_0(x), &\text{ in }\Omega
	\end{aligned}\right.
\end{equation}
with the initial condition is given by $w_0(x)=u_0(x)-u_g(x,0)$. Since $w$ satisfies homogeneous Dirichlet boundary condition, we can still use the Laplace eigenfunctions with homogeneous Dirichlet boundary conditions to compute the coefficient $\bm{\widetilde{\beta}}^n$ of $w$. Using $w$, we can apply Algorithm \ref{alg:no} to train the neural operator $\mathcal{N}$. For the input to the neural operator $\mathcal{G}$, the coefficients of $u_g$ can either be included as part of the input or excluded for training.
\subsection{Neumann Boundary Condition}
For the Neumann boundary condition
\begin{equation}
	{\partial u\over \partial n}=g.
\end{equation}
We also consider harmonic extension of the boundary to obtain $u_g \in L^2_0(\Omega) \cap H^1(\Omega)$,  which satisfies:
\begin{equation}
	\left\{
	\begin{aligned}
		-\Delta u_g &= 0, &\text{in}~\Omega \times (0,T)\\
		{\partial u\over \partial n} &=g,&\text{on}~\partial\Omega\times(0,T)
	\end{aligned}
	\right.
\end{equation}
This harmonic extension can be derived using the finite element method. Using $u_g$, we reformulate the original problem in terms of $w:=u-u_g$ which satisfies:
\begin{equation}
	\left\{
	\begin{aligned}
		w_t-D\Delta w &= \mathcal{F}(w+u_g)-(u_g)_t, & \text{ in }\Omega\times [0,T]\\
		w(x,0) &= w_0(x), &\text{ in }\Omega
	\end{aligned}\right.
\end{equation}
with the initial condition is given by $w_0(x)=u_0(x)-u_g(x,0)$. Since $w$ satisfies homogeneous Neumann boundary condition, we can still use the Laplace eigenfunctions  in $L^2_0(\Omega)$ with homogeneous Neumann boundary conditions to compute the coefficient $\bm{\widetilde{\beta}}^n$ of $w$. Using $w$, we can apply Algorithm \ref{alg:no} to train the neural operator $\mathcal{N}$. For the input to the neural operator $\mathcal{G}$, the coefficients of $u_g$ can either be included as part of the input or excluded for training.
\section{Details for Numerical Tests}\label{appdendix:test}
\subsection{Datasets and Complementary Results}
In this section, we provide the details of datasets and some complementary results for the equation we used in Section \ref{sec:experiment}.

\subsubsection{KPP-Fisher Equation}
Recall the 1D KPP-Fisher equation on the unit interval:
\begin{equation}
	\left\{	\begin{aligned}
		{\partial u\over \partial t}-{\partial^2 u\over \partial x^2}  &= u(1 -u),  &x\in (0,1),t\in (0,T]\\
		u(x,0) &= u_0(x), &x\in (0,1)
	\end{aligned}\right.
\end{equation}
The initial condition $u_0(x)$  is generated according to $u_0\sim\mu $ where $\mu=\mathcal{N}(0,49(-\Delta+7I)^{-2.5})$ with corresponding Dirichlet boundary conditions. We solve the governing equation using an explicit Euler method, employing a piecewise linear finite element method for spatial discretization. The computational domain is discretized on a spatial mesh with a resolution of $2^{10}=1024$ and  the time step is set to $\tau=0.01$. For the dataset used to train the neural network, data is collected over 10 time steps, using 100 different solutions. 
For the evaluation and comparison of predicted results, we include 30 time steps.

\jdall{
For the extended diffusion coefficient tests (Cases1 and2), the data are generated following the same procedure as described above.
}

\subsubsection{Allen-Cahn Equation}
The Allen-Cahn equation is given by:
\begin{equation}
	\left\{	\begin{aligned}
		{\partial u\over \partial t} - \varepsilon^2 {\partial^2 u\over\partial x^2} &= W'(u),  &x\in (0,2\pi),t\in (0,T]\\
		u(x,0) &= u_0(x),&t\in (0,T)\\
	\end{aligned}\right.
\end{equation}
We set $\varepsilon=0.1$. The initial condition  $u_0 $ is sampled from the distribution $\mu=\mathcal{N}(0,49(-\Delta+7I)^{-2.5})$ with pure Neumann boundary conditions. An explicit Euler method, with a piecewise linear finite element method for spatial discretization is used. The computational domain is discretized on a spatial mesh with a resolution of $2^{10}=1024$ and  the time step is set to $\tau=0.05$. For the dataset used to train the neural network, data is collected over 20 time steps, using 100 different solutions. For the evaluation and comparison of predicted results, we include 60 time steps. The training error is given in Table \ref{tab:ac}.

\begin{table}[!htbp]
	\centering
	\caption{Errors of the Allen-Cahn equation.}
	\begin{tabular}{ccc}
		\hline
		$E_{L^2}$ &  $E_{Res}$  &$E_{Nonlinear}$\\
		\hline
		6.58e-04&2.45e-03&6.59e-03\\
		\hline
	\end{tabular}
	\label{tab:ac}
\end{table}

\subsubsection{Gray-Scott Equation}
The Gray-Scott model is a reaction--diffusion system with two variables,  $A$  and  $S$ , described by the following equations:
\begin{equation}
	\left\{	\begin{aligned}
		A_t-D_A\Delta A &= SA^2-(\mu+\rho)A, &x \in (0,2\pi)^d,t\in(0,T)\\
		S_t-D_S\Delta S &= -SA^2+\rho(1-S), &x\in (0,2\pi)^d,t\in(0,T)\\
		A(x,0) &= A_0(x)  &x \in (0,2\pi)^d\\
		S(x,0) &= S_0(x)  &x \in (0,2\pi)^d\\
	\end{aligned}\right.
\end{equation}
where  the parameters are set as $D_A=2.5\times 10^{-4}$, $D_S=5\times10^{-4}$, $\rho=0.04$ and $\mu=0.065$. The initial condition  $A_0 $ and $S_0$ are sampled from the distribution $\mu=\mathcal{N}(0,49(-\Delta+7I)^{-2.5})$ with pure Neumann boundary conditions. An explicit Euler method, with a piecewise linear finite element method for discretization is used. The computational domain is discretized on a spatial mesh with a resolution of $2^{10}=1024$ for 1D problems and $256\times 256$ for 2D problems. The time step is set to $\tau=0.1$. For the dataset used to train the neural network, data is collected over 20 time steps, using 100 different solutions. For the evaluation and comparison of predicted results, we include 60 time steps. Note that this equation involves two variables, so we include the projection coefficients of both $A$ and $S$ as inputs to the neural network.

\subsubsection{Schr\"odinger Equation}
The Schr\"odinger equation is given by:
\begin{equation}
	\left\{\begin{aligned}
		&u_t	-\Delta u +\alpha |u|^2u+Vu=\lambda u, &x\in (-8,8)^2,t\in(0,T)\\
		& u(x,0)=u_0(x) &x\in (-8,8)^2\\
	\end{aligned}
	\right.
\end{equation}
where $\alpha=1600$, $\lambda=15.87$ is the eigenvalue of the steady problem and the potential is defined as:
\begin{equation}
	V(x,y) = 100(\sin^2({\pi x\over 4})+\sin^2({\pi y\over 4}))+x^2+y^2.
\end{equation}
The initial condition  $u_0 $ is sampled from the distribution $\mu=\mathcal{N}(0,49(-\Delta+7I)^{-2.5})$ with homogeneous Dirichlet boundary conditions. We employ an explicit Euler method for time integration, combined with a piecewise linear finite element method for spatial discretization.The computational domain is discretized on a spatial mesh with a resolution of $256\times 256$ . The time step is set to $\tau=0.001$. For the dataset used to train the neural network, data is collected over 20 time steps, using 100 different solutions. For the evaluation and comparison of predicted results, we include 60 time steps. The training error is given in \ref{tab:Sch}.

\begin{table}[!htbp]
	\centering
	\caption{Errors of the Schr\"odinger equation. }
	\begin{tabular}{ccc}
		\hline
		$E_{L^2}$ &  $E_{Res}$  &$E_{Nonlinear}$\\
		\hline
		9.91e-04&9.37e-03& 1.79e-02\\
		\hline
	\end{tabular}
	\label{tab:Sch}
\end{table}

\subsubsection{Alzheimer's Disease Evolution}
We use the gray values of MRI scans as both the data and the input images. A brain slice at a specific $z$-location is selected to define the irregular domain for training and transfer learning, for the time series there are 3 times steps. The training loss curve is displayed in Figure \ref{fig:ctloss}.

\begin{figure}[!htbp]
	\centering
	\includegraphics[width=.4\textwidth]{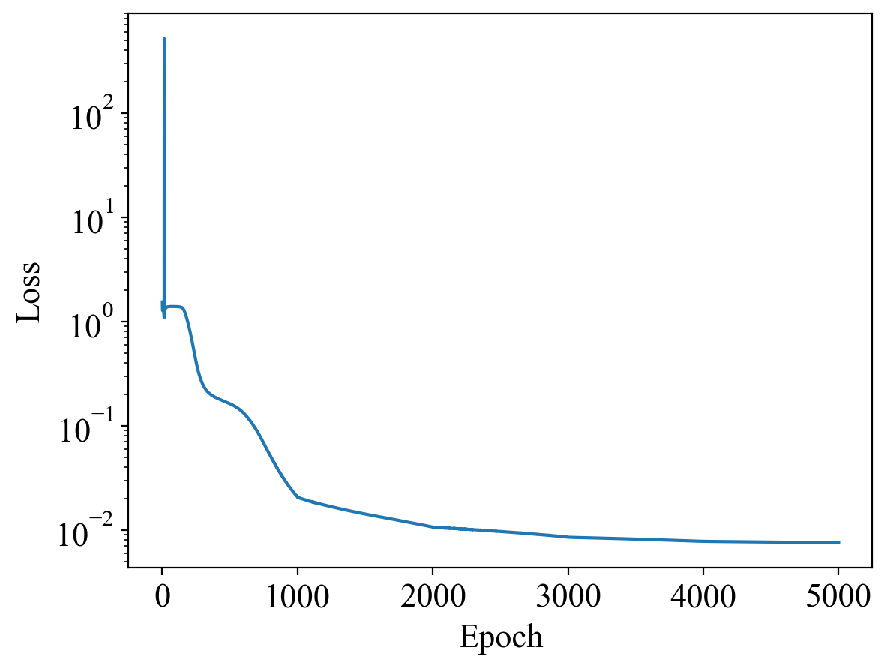}
	\caption{Training loss curve on Alzheimer's disease evolution.}
	\label{fig:ctloss}
\end{figure}

\subsection{Training Details}
We use the Adam optimization method with a full data batch for 5000 epochs, starting with an initial learning rate of $10^{-3}$, which is reduced by a factor of 0.25 every 1000 epochs.

\subsection{Neural Network Architecture}
We utilize a fully connected neural network for $\mathcal{N}$  employing the ReLU activation function. The architecture details, including the number of hidden layers and eigenfunctions used in different tests, are summarized in Table \ref{tab:arc}.

\begin{table}[!htbp]
	\centering
	\caption{Neural Network Architecture.}
    \jdall{
	\begin{tabular}{c|cc}
		\hline
		Problem 	& \#Eigen function & Neural Network\\
		\hline
		KPP-Fisher & 64 & 64-1000-1000-64 \\
        KPP-Fisher (Case 1) & 64 & 64-1000-1000-64 \\
        KPP-Fisher (Case 2) & 256 & 256-1000-1000-256 \\
		Allen-Cahn & 64 & 64-1000-64 \\
		1D Gray-Scott & 64 & 128-1000-1000-128\\
		2D Gray-Scott & $48^2$ & $2*48^2$-1000-1000-1000-$2*48^2$\\
		Schr\"odinger & $64^2$ & $64^2$-1000-1000-1000-$64^2$\\
		Alzheimer's disease evolution& $16^2$ &  $32^2$-16-$32^2$\\
		\hline
	\end{tabular}
    }
	\label{tab:arc}
\end{table} 

\bibliographystyle{elsarticle-num}
\bibliography{main}

\end{document}